\documentclass[preprint,12pt,table]{elsarticle}

\usepackage{graphicx}

\usepackage{amssymb}

\usepackage{amsthm}
\usepackage{amsmath}
\usepackage{mathpazo}
\usepackage{tikz}
\usepackage{url}
\usepackage[a4paper, total={7in, 10in}]{geometry}
\usepackage{hhline}
\usepackage[ruled, linesnumbered, vlined]{algorithm2e}
\usepackage{color,soul}
\usepackage{rotating}

\usepackage{mathtools}

\DeclareMathOperator{\spe}{sp}
\DeclareMathOperator{\pdf}{pdf}
\DeclareMathOperator{\hgt}{hgt}
\DeclareMathOperator{\hyb}{hyb}
\DeclareMathOperator{\ext}{ext}
\DeclareMathOperator{\Ta}{Taxa}

\makeatletter
\def\ps@pprintTitle{%
 \let\@oddhead\@empty
 \let\@evenhead\@empty
  \def\@oddfoot{\centerline{\thepage}}%
 \let\@evenfoot\@oddfoot}
\makeatother

\usepackage[symbol]{footmisc}

\newtheorem{theorem}{Theorem}[section]
\newtheorem{lemma}[theorem]{Lemma}

\theoremstyle{definition}

\newcommand{\T}{\mathcal{T}}

\journal{}

\begin{document}

\begin{frontmatter}

\title{Comparing the topology of phylogenetic network generators}

\author{Remie Janssen{\footnotesize$^1$}}
\author{Pengyu Liu{\footnotesize$^2$}\footnote{To whom correspondence should be addressed; e-mail: pengyu\_liu@sfu.ca.}}

\address{$^1$Delft Institute of Applied Mathematics, Delft University of Technology, Delft, The Netherlands}
\address{$^2$Department of Mathematics, Simon Fraser University, Burnaby, Canada}

\begin{abstract}
Phylogenetic networks represent evolutionary history of species and can record natural reticulate evolutionary processes such as horizontal gene transfer and gene recombination. This makes phylogenetic networks a more comprehensive representation of evolutionary history compared to phylogenetic trees. Stochastic processes for generating random trees or networks are important tools in evolutionary analysis, especially in phylogeny reconstruction where they can be utilized for validation or serve as priors for Bayesian methods. However, as more network generators are developed, there is a lack of discussion or comparison for different generators. To bridge this gap, we compare a set of phylogenetic network generators by profiling topological summary statistics of the generated networks over the number of reticulations and comparing the topological profiles.

\end{abstract}

\end{frontmatter}


\section{Introduction}
\label{S1}

A phylogenetic network is a directed acyclic graph that represents the evolutionary history of a set of species \cite{huson2010phylogenetic}. With the recent advancements in the theory of phylogenetic networks and the accumulation of evolutionary data, it is becoming increasingly viable and common to represent the evolutionary hierarchical relations with phylogenetic networks instead of phylogenetic trees, which cannot record reticulate evolutionary processes such as horizontal gene transfer (also known as lateral gene transfer), hybrid speciation, hybrid introgression, and recombination \cite{elworth2019advances,blair2020phylogenetic,blais2021past}.


In evolutionary analysis, the stochastic process that generates random trees and random networks play an important role, especially in the reconstruction of phylogenies. Similar to phylogenetic tree reconstruction, there are two main types of methods for rebuilding evolutionary histories as phylogenetic networks: combinatorial methods and model based methods. Combinatorial methods are based on, for example, distances between taxa \cite{bordewich2016determining,bordewich2018recovering,van2020reconstructibility} or gene trees \cite{baroni2005bounding,linz2019attaching,whidden2013fixed,van2019practical,van2019polynomial}, while model based methods involve Bayesian statistics \cite{zhang2018bayesian,zhu2018bayesian} or likelihood calculations \cite{yu2015maximum,solis2016inferring} in models of sequence and lineage evolution. For both types of methods, random phylogenetic networks are needed in validating the reconstruction methods. Moreover, for Bayesian methods, the stochastic processes used to generate random networks also serve as priors \cite{zhang2018bayesian}. In this paper, we call the stochastic processes that generate random phylogenetic networks the {\em network generators}. Note that this is different from the generator of a network defined in \cite{van2009constructing}, which encodes the underlying reticulate structure of a phylogenetic network. 



There are a few approaches that existing network generators usually take to constructing random networks. One approach is to alter a tree generator and add hybrid edges, for example, network generators \cite{van2019practical,pons2019generation} that are based on the Yule-Harding model \cite{harding1971probabilities} or on a continuous-time Markov model of speciation to a birth-hybridization Markov model \cite{zhang2018bayesian,morin2006netgen}. The hybrid edges can either be added during the process of building a tree or after a tree is generated. For these methods, there is also an option to add the hybrid edges in a way that favours short-distance (local) hybrid edges \cite{morin2006netgen,pons2019generation}. 

Other approaches that do not rely on tree generating processes include assembling a network directly from degree constraints \cite{zhang2016tree} and sampling a network in the network space by random local surgeries, e.g., rearrangement moves which move one endpoint of an edge \cite{gambette2017rearrangement,janssen2021rearranging}. The latter is used in Bayesian methods for sampling posteriors \cite{zhang2018bayesian,zhu2018bayesian}.
In order to present a more comprehensive comparison of network generators, we include two new network generators in our analysis, in addition to the ones already available in the literature. The first is based on the beta-splitting model for trees \cite{aldous1996probability}, which we extend to networks by adding edges to a generated tree in various ways. The second is based on a birth-death process \cite{heath2008taxon} which we extend with distance dependent hybrid speciation events.

For phylogenetic trees, it is well discussed which stochastic process generates phylogenetic trees that resemble trees reconstructed from data \cite{blum2006random}, and random tree generators are well compared from different points of view \cite{steel2016phylogeny,Liu20202,choi2020cherry,bienvenu2020revisiting}. However, as increasingly many network generators are developed, there is a lack of discussion about the differences and similarities of generated random networks. To bridge this gap, we examine and compare network generators by profiling topological summary statistics of the generated networks over the number of reticulations in the networks and comparing the topological profiles. The topological summary statistics used for the comparison include network balance, number of cherries, number of reticulated cherries, number, size and level of blobs as well as network class. Moreover, we generalize the recent polynomial representation for trees \cite{LIU20211} to networks and leverage the high-resolution polynomial representation in comparing small phylogenetic networks. 

We show that generators based on different types of approaches can generate networks with similar topological profiles. If a network generator is based on a tree generator, the topological characteristics are dominated by the topology of the underlying trees when the number of reticulations is small; and as the number of reticulations increases, the topological characteristics may diverge with respect to the integrated reticulation structures and the locality of added reticulations.

\section{Methods}
\label{S2}

\subsection{Topological profiles}

A {\em phylogenetic network} or simply a {\em network} is a rooted directed acyclic graph whose out-degree zero nodes ({\em leaf nodes}) have in-degree one. We study binary networks, that is, every internal node in the network has in-degree one and out-degree two ({\em tree nodes}) or in-degree two and out-degree one ({\em reticulation nodes}). In a phylogenetic network, a reticulation event is represented by a reticulation node and the edges pointing to the reticulation node which are called \emph{hybrid edges}. A reticulation node together with the hybrid edges incident to it is a {\em reticulation structure}.


The {\em topological profile} of a network generator is defined to be the topological summary statistics of generated networks as a function of number of reticulations. In other words, the topological profile of a network generator describes how the topological summary statistics change as the number of reticulations increases. We choose commonly used topological characteristics of phylogenetic networks for our comparison. These characteristics include network balance, number of cherries and reticulated cherries, size and level of blobs as well as proportions of generated networks that belong to various network classes.



\subsubsection{Network balance}


The balance of phylogenetic trees is a well studied topological characteristic. The balance for phylogenetic networks is recently revisited in \cite{bienvenu2020revisiting}, where the $B_2$ balance for phylogenetic networks \cite{shao1990tree} is reintroduced. The $B_2$ balance is an index based on the entropy of probability distribution on leaf nodes. Specifically, let $N$ be a network with the set of leaf nodes denoted by $L$ and $p_l$ be the probability of a directed random walk starting from the root and ending at the leaf node $l$, then the $B_2$ balance of $N$ can be computed by $B_2(N)=-\sum_{l\in L} p_l \log_2(p_l)$.

The maximal value of $B_2$ balance for a network with $n$ leaf nodes is $\log_2(n)$, which corresponds to the most balanced case where every leaf node has probability $1/n$. For our comparison, we use the normalized $B_2$ balance, which is computed by dividing the $B_2$ balance a of network with $n$ leaf nodes by $\log_2(n)$.


\subsubsection{Cherries and reticulated cherries}
The number of small subgraphs such as cherries and pitchforks is another well-studied topological characteristics for phylogenetic trees \cite{choi2020cherry}. A {\em cherry} of a network is a subgraph consisting of two leaf nodes $x$ and $y$ that share a common parent node. For our comparison of phylogentic networks, in addition to the number of cherries, we also compute the number of {\em reticulated cherries} which are subgraphs consisting of two leaf nodes $x$ and $y$, a reticulation node $w$ as the parent node of $x$ (or $y$), and the common parent of $y$ and $w$ (or $x$ and $w$) \cite{bordewich2016determining}. 


\subsubsection{Blobs}

A blob in a network is a maximal connected subgraph without cut edges and consisting of at least two nodes \cite{gusfield2005fundamental}. 
The {\em size} of a blob in a network is the number of nodes in the blob; the {\em level} of a blob is the number of reticulation nodes in the blob; and the {\em level} of a network is the maximal level over all its blobs. 


\subsubsection{Network classes}
To overcome the obstacles faced in extending the mathematical theory of phylogenetic trees to phylogenetic networks, some subclasses of phylogenetic networks were introduced. These network classes include tree-child networks, stack-free networks, orchard networks, and tree-based networks. We examine the proportion of generated networks that belong to each of the classes as another topological characteristic in the profile of a generator. 


A network is {\em tree-based} if it has an embedded spanning tree whose leaf nodes are exactly the leaf nodes of the network \cite{francis2015phylogenetic}. A network is {\em orchard} if it can be reduced to a network consisting of a single vertex by a sequence of cherry reduction operations \cite{janssen2021cherry,erdHos2019class}. A stack in a phylogenetic network is a pair of adjacent reticulation nodes. A network is {\em stack-free} if there are no stacks in the network \cite{semple2018phylogenetic}. A network is {\em tree-child} if each internal node of the network has at least one child that is not a reticulation node \cite{cardona2008comparison}. 





For binary networks, a tree-child network is also a stack-free network as well as an orchard network \cite{janssen2021cherry}.
Moreover, every stack-free network or orchard network is also tree-based \cite{zhang2016tree,huber2019rooting,van2020unifying}.


\subsection{Polynomial comparison}

A polynomial representation for phylogenetic trees was introduced in \cite{Liu20202}. The polynomial representation is an isomorphic invariant for trees, that is, two trees are isomorphic if and only if they have the same polynomial representation. Moreover, the polynomials for rooted trees are irreducible, that is, they cannot be factored as a multiplication of smaller polynomials \cite{LIU20211}. We leverage these results to extend the polynomial representation for further comparing phylogenetic network generators. 

A natural way to connect phylogenetic networks with trees is by considering their sets of spanning trees. Let $N$ be a phylogenetic network and $\T_N$ be the deck of rooted spanning trees of $N$, that is, $\T_N$ may contain identical elements. It is natural to represent $N$ with the polynomial $P(N,x,y)=\prod_{T \in \T_N}P(T,x,y)$. It is unknown in general whether graphs can be reconstructed from their decks of spanning trees. However, it has been proved that two tree-child networks are isomorphic if and only if they have the same set of embedded spanning trees \cite{francis2018identifiability}. If we compute the polynomial for tree-child networks based on the set of embedded spanning trees, it is immediate that two tree-child networks are isomorphic if and only if they have the same polynomial. However, although not all networks are tree-child networks, it is still useful to compute their polynomials for comparison.

To compare the networks, We use the polynomial distance defined in \cite{Liu20202}, that is, the Canberra distance \cite{Lance} between the polynomial coefficient vectors of the networks. Then we use multi-dimensional scaling \cite{mds} to visualize the distances, and we deploy k-medoids clustering \cite{kmedoids} on the distances to determine if networks generated by different generators can be distinguished by the polynomial distance. 






\subsection{Network generators}

Phylogenetic network generators typically alter tree generators by introducing reticulation event to the process. 
An important phylogenetic tree generator is the beta-splitting model which is a generalization of the Yule model and the PDA model \cite{aldous1996probability,blum2006random}. 

For our comparison, we select the LGT network generator \cite{pons2019generation} and the ZODS
network generator \cite{zhang2018bayesian} as representations of network generators based on the Yule model (i.e., beta-splitting model with $\beta=0$). We compare these network generators to a new network generator which is an extension of Heath's tree generator \cite{heath2008taxon}. Another network generator we choose is the NTK network generator \cite{zhang2016tree} which is described not in terms of a branching process, but it utilizes the degree constraints of networks for direct sampling. Moreover, for comparison, we also introduce two additional network generators as references. 
The first is an extension of the beta-splitting tree generator. Instead of adding reticulations in the process, the beta-splitting network generator adds reticulations to generated trees in various ways. The second is the MCMC network sampler which is commonly used in Bayesian methods for sampling posteriors, but we employ it as a standalone generator.

For network generators that are based on a phylogenetic tree generator, we will discuss three types of reticulation structure integrated in the branching process. See Figure \ref{fig1}. For a pair of leaf nodes $x$ and $y$ of the network under generation at the current state, the {\em n-type} reticulation adds an edge from $x$ to a new leaf node $x'$, an edge from $y$ to a new leaf node $y'$ and a horizontal edge from $x$ to $y$ (or $y$ to $x$); the {\em y-type} reticulation glues the leaf node $x$ and $y$ and produces a new leaf node $z$; the {\em m-type} reticulation branches at $x$ and $y$ and glues a pair of new leaf nodes producing three new leaf nodes $x'$, $y'$ and $z$. The n-type reticulation structure has a natural interpretation as an HGT or an hybrid introgression event, and the m-type can be interpreted as a hybrid speciation event.

\begin{figure}[tbh]
\begin{center}
    
    \begin{tikzpicture}[scale=0.8, every node/.style={scale=0.8}]

    \definecolor{clr1}{RGB}{0,114,189}
    \definecolor{clr2}{RGB}{217,83,25}
    \definecolor{clr3}{RGB}{237,177,32}
    \definecolor{clr4}{RGB}{126,47,142}
    \definecolor{clr5}{RGB}{119,172,48}
    
    \node[shape=circle, draw=black, fill = black, scale = 0.5] (X1) at (0,0) {};
    \node[shape=circle, draw=black, fill = black, scale = 0.5] (Y1) at (2,0) {};
    \node[shape=circle, draw=black, fill = black, scale = 0.5] (X2) at (0,-1) {};
    \node[shape=circle, draw=black, fill = black, scale = 0.5] (Y2) at (2,-1) {};
    
    \path[->] (X1) edge (X2);
    \path[->] (Y1) edge (Y2);
    
    \node [rectangle,xshift=0cm] at (-0.3,-1) (return) {$x$};
    \node [rectangle,xshift=0cm] at (2.3,-1) (return) {$y$};
    
    \node [rectangle,xshift=0cm] at (1,-4.25) (return) {leaf nodes};

    \node[shape=circle, draw=black, fill = black, scale = 0.5] (X3) at (5,0) {};
    \node[shape=circle, draw=black, fill = black, scale = 0.5] (Y3) at (7,0) {};
    \node[shape=circle, draw=black, fill = black, scale = 0.5] (X4) at (5,-3) {};
    \node[shape=circle, draw=black, fill = black, scale = 0.5] (Y4) at (7,-3) {};
    \node[shape=circle, draw=black, fill = black, scale = 0.5] (Z1) at (5,-1) {};
    \node[shape=circle, draw=clr2, fill = clr2, scale = 0.5] (Z2) at (7,-1) {};
    
    \path[->] (X3) edge (Z1);
    \path[->] (Z1) edge (X4);
    \path[->, draw=clr2] (Y3) edge (Z2);
    \path[->] (Z2) edge (Y4);
    \path[->, draw=clr2] (Z1) edge (Z2);

    \node [rectangle,xshift=0cm] at (4.7,-1) (return) {$x$};
    \node [rectangle,xshift=0cm] at (7.3,-1) (return) {$y$};
    
        \node [rectangle,xshift=0cm] at (4.7,-2.95) (return) {$x'$};
    \node [rectangle,xshift=0cm] at (7.3,-2.95) (return) {$y'$};

    \node [rectangle,xshift=0cm] at (6,-4.25) (return) {n-type};

    \node[shape=circle, draw=black, fill = black, scale = 0.5] (X5) at (15,0){};
    \node[shape=circle, draw=black, fill = black, scale = 0.5] (Y5) at (17,0) {};
    \node[shape=circle, draw=black, fill = black, scale = 0.5] (X6) at (15,-3) {};
    \node[shape=circle, draw=black, fill = black, scale = 0.5] (Y6) at (17,-3) {};
    \node[shape=circle, draw=black, fill = black, scale = 0.5] (Z3) at (15,-1) {};
    \node[shape=circle, draw=black, fill = black, scale = 0.5] (Z4) at (17,-1) {};
    \node[shape=circle, draw=clr2, fill = clr2, scale = 0.5] (Z5) at (16,-2) {};
    \node[shape=circle, draw=black, fill = black, scale = 0.5] (Z6) at (16,-3) {};
    
    \path[->] (X5) edge (Z3);
    \path[->] (Z3) edge (X6);
    \path[->] (Y5) edge (Z4);
    \path[->] (Z4) edge (Y6);
    \path[->, draw=clr2] (Z3) edge (Z5);
    \path[->, draw=clr2] (Z4) edge (Z5);
    \path[->] (Z5) edge (Z6);
    
    \node [rectangle,xshift=0cm] at (16,-4.25) (return) {m-type};
    
    \node [rectangle,xshift=0cm] at (14.7,-1) (return) {$x$};
    \node [rectangle,xshift=0cm] at (17.3,-1) (return) {$y$};
    
\node [rectangle,xshift=0cm] at (14.7,-2.95) (return) {$x'$};
    \node [rectangle,xshift=0cm] at (17.3,-2.95) (return) {$y'$};
    \node [rectangle,xshift=0cm] at (15.7,-3) (return) {$z$};
    
    
    \node[shape=circle, draw=black, fill = black, scale = 0.5] (X8) at (10,0) {};
    \node[shape=circle, draw=black, fill = black, scale = 0.5] (Y8) at (12,0) {};
    \node[shape=circle, draw=clr2, fill = clr2, scale = 0.025] (X7) at (10,-1) {};
    \node[shape=circle, draw=clr2, fill = clr2, scale = 0.025] (Y7) at (12,-1) {};
    \node[shape=circle, draw=clr2, fill = clr2, scale = 0.5] (Z7) at (11,-2) {};
    \node[shape=circle, draw=black, fill = black, scale = 0.5] (Z8) at (11,-3) {};
    
    \path[-, draw=clr2] (X8) edge (X7);
    \path[-, draw=clr2] (Y8) edge (Y7);
    \path[->, draw=clr2] (X7) edge (Z7);
    \path[->, draw=clr2] (Y7) edge (Z7);
    \path[->] (Z7) edge (Z8);
    
    \node [rectangle,xshift=0cm] at (9.7,-1) (return) {$x$};
    \node [rectangle,xshift=0cm] at (12.3,-1) (return) {$y$};
    \node [rectangle,xshift=0cm] at (10.7,-3) (return) {$z$};

    \node [rectangle,xshift=0cm] at (11,-4.25) (return) {y-type};

\end{tikzpicture}
    
    \caption{Different types of reticulation structures, where reticulation nodes and hybrid edges are in red.}
    \label{fig1}
\end{center}
\end{figure}
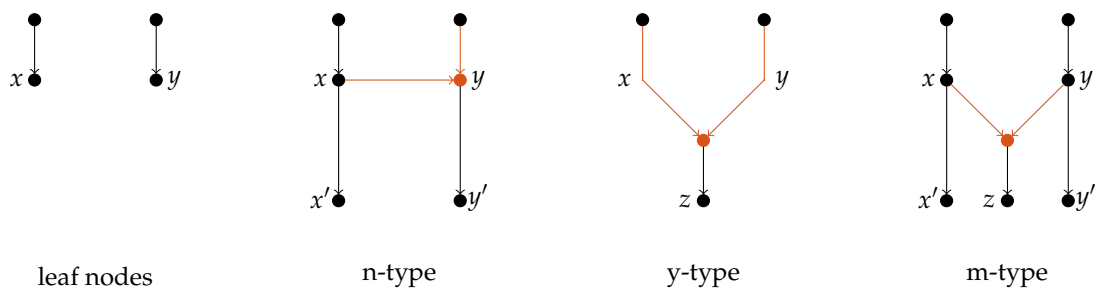

\subsubsection{LGT network generator}

The LGT network generator is a network generator designed to model the lateral gene transfer (LGT) events. The generator is introduced in \cite{pons2019generation}, and, from its description, the speciation events are designed as in the Yule model, that is, a leaf node at the current state is chosen uniformly at random for branching. The LGT network generator uses n-type reticulations and has two parameters governing the introduction of reticulation events: a parameter $\alpha\in[0,1]$ controls the probability of the next event being a reticulation event instead of a speciation event, and a parameter $\gamma\geq 0$ regulates the probability of a reticulation event being within an existing blob, where smaller $\gamma$ corresponds to higher probability of a local reticulation event. Hence, if the number of events is fixed, the parameter $\alpha$ influences the number of reticulations in the network, and the parameter $\gamma$ influences the level of the network. For our comparison, we modified the code of the generator so that we have networks with the same number of reticulations and number of leaf nodes, making the $\alpha$ parameter irrelevant for our comparison.

\subsubsection{ZODS network generator}

The ZODS network generator is a birth-hybridization process introduced in \cite{zhang2018bayesian} as a prior in the Bayesian method. The generator uses y-type reticulations and has a constant speciation rate $\lambda$ and a constant hybridization rate $\nu$, that is, at a step when there are $k$ extant leaf nodes in a network, the speciation rate is $\lambda k$ and hybridization rate $k(k-1)\nu/2$. Besides the reticulation type, another difference between the ZODS network generator and the LGT network generator is that the hybridization rate. The hybridization rate in the ZODS network generator depends quadratically on the number of leaf nodes, whereas in the LGT network generator, the hybridization rate depends linearly on the number of leaf nodes.

\subsubsection{NTK network generator}

A binary phylogenetic network has the degree constraints that a tree node has in-degree one and out-degree two, a leaf node has in-degree one and out-degree zero, and a reticulation node has in-degree two and out-degree one. The NTK network generator introduced in \cite{zhang2016tree} does not rely on a branching process but utilizes these degree constraints and generates networks as follows. To generate a random network of a certain size, the generator selects uniformly at random $n$ nodes as leaf nodes and $r$ nodes as reticulation nodes from $2(n+r)-1$ nodes, and randomly inserting edges between a pair of nodes so that the degree constraints are met.


\subsubsection{Beta-splitting network generator}

The beta-splitting network generator is based on the beta-splitting model for trees \cite{aldous1996probability,blum2006random,sainudiin2016beta}. The beta-splitting model for trees is a generalization of the Yule model ($\beta=0$) and the PDA model ($\beta = -1.5$). Here, we introduce three methods to create a network from a generated tree by repeatedly adding edges to the tree. The first method picks two edges uniformly at random, subdivides these two edges with two vertices respectively and adds another new edge between these two vertices with a random direction that does not introduce cycles (Algorithm~\ref{alg:AddEdgesUniform}); the second chooses edges more locally governed by a single parameter $P_{\text{stop}}\in[0,1]$ that stochastically determines how local the cycles are, a high $P_{\text{stop}}$ leads to more local reticulations (Algorithm~\ref{alg:AddEdgeLocal}); the third picks two leaf nodes, subdivides their incoming edges with two vertices and adds a horizontal edge between the two vertices with a random direction (Algorithm~\ref{alg:AddEdgesHorizontal}). We call the process of generating networks based on the PDA trees in these ways the {\em PDA network generator}, and the generated networks the {\em PDA networks}. 

\subsubsection{MCMC network sampler}
The Markov-chain-Monte-Carlo (MCMC) network sampler performs a random walk through the space of phylogenetic networks by applying rSPR moves, small changes to the network that move one endpoint of an edge to another edge \cite{gambette2017rearrangement}. By selecting networks after sufficiently many random rSPR moves, the MCMC network sampler is able to sample (almost) uniformly from the set of (internally labelled) networks, as this space is connected \cite{janssen2021rearranging}. The MCMC network sampler lies at the basis of a common method in Bayesian analysis for sampling posteriors, as in PhyloNet \cite{wen2016bayesian} and SpeciesNetwork \cite{zhang2018bayesian}.


\subsubsection{Heath network generator}

Finally, we present a new network generator based on the branching process introduced in \cite{heath2008taxon}, which uses auto-correlated speciation and extinction rates. For our analysis, we set the extinction rate to zero. The branching process is extended by adding hybrid speciation events, i.e., we integrate the m-type reticulation in this model. We introduce a new distance dependent hybridization rate $f_{hyb}(d)$ between two leaf nodes at the current state, which depends on the weighted distance $d$ between the leaf nodes, where $d$ is the sum of the lengths of all up-down paths between the two leaves, weighted by the probability of the up-down path, i.e., the product of the edge-probabilities on the path. The function $f_{hyb}(d)$ is piecewise linear, and it is linearly decreasing in $d$ (for $h_l\leq d \leq h_r$) with a minimum and maximum rate $h_{lr}$ and $h_{rr}$.
\[
f_{hyb}(d) = \begin{cases} h_{lr}                                           &\mbox{if } d \leq h_l \\
                           h_{lr}+\frac{(h_{rr}-h_{lr})(d-h_l)}{h_r-h_l}    &\mbox{if } h_l<d<h_r  \\
                           h_{rr}                                           &\mbox{if } d \geq h_r \end{cases} .\]

See \ref{sec:HeathUpdateDistance} for more details of the generator, and for an efficient method to update the weighted distance $d$ after several types of evolutionary events. Using these, our distance dependent extension is also viable for other tree generators than the Heath tree generator.

\subsection{Data}

To compare the network generators, we generate a large number of networks with each of them, varying the parameters if there are any. To keep the comparison straightforward, we generate networks with 100 leaf nodes for this purpose, and we vary the number of reticulations. When using the polynomial to compare networks, we use smaller networks with 10 leaves. 

\paragraph{LGT networks}

For the LGT network generator, we use three different parameter values ($0.01$, $0.1$, and $1.0$) for the parameter $\gamma$ that regulated the locality of the reticulations. The parameter $\alpha$ becomes irrelevant as we fix the number of reticulation nodes and leaf nodes as mentioned before. For each of these parameter values, we generate 500 networks with $r=1$, 2, 3, 4, 5, 10, 20, 50, 100, 200, 500 reticulations. This results in a set of $16500$ networks.

\paragraph{ZODS networks}

Setting the number of leaf nodes and reticulation nodes, the parameters become irrelevant as we have argued before. Hence, we generate networks with a set number of reticulations using a reticulation rate that gives a large probability to generate such a network. We used the parameters $\nu=0.0005,0.001,0.002$ to create 500 networks for each reticulation number $r=1,\ldots,5$; we used $\nu=0.005$ to create 500 networks with 10 and 20 reticulations, $\nu=0.01$ for 500 networks with 20 and 50 reticulations, $\nu=0.015$ for 500 networks with 50 and 100 reticulations, and, finally, $\nu=0.02$ for 500 networks with 100 and 200 reticulations. Hence, the full set of networks totalled $23\cdot 500 = 11500$ networks.

\paragraph{NTK networks}
For $r=1,2,3,4,5,10,20,50,100,200$ reticulations, we downloaded 500 networks with 100 leaf nodes from the webpage applet. This generator has no further parameters, so we have a total of $5000$ networks for this generator.

\paragraph{Beta-splitting networks}
We used the beta values $0.0$, $-0.5$, $-1.0$, and $-1.5$. For each of these values, we created 50 trees and added $r=1,2,3,4,5,10,20,50,100,200,500$ reticulations to these trees to create 10 networks with each of the following methods: Uniformly 
, Horizontally 
, and locally 
with stopping probability $P_{\text{stop}}\in\{0.2, 0.02, 0.002\}$. This results in a total of $4\cdot 50\cdot 11 \cdot 10\cdot 5 = 110000$ networks.

\paragraph{MCMC sampled networks}
We sample networks with $r=1,2,3,4,5,10,20,50,100,200,500$ reticulations, starting with a network sampled by the beta-splitting network generator with $\beta=-1.0,-1.5$ adding edges uniformly. 
For each of these parameter combinations, we have sampled $500$ networks. This results in a set of $1000$ networks for each reticulation number, and $11000$ networks in total.

Each sample is taken after $5000$ proposals. Because about half of the proposals were rejected as invalid rSPR moves, we took about $2500$ steps through the space of networks between each sample. As the diameter of the space of networks with $n$ leaf nodes and $r$ reticulation nodes is at most $2n+3r-2$; see \cite{janssen2018exploring}. This is reasonable to make the samples independent for $n=100$ and $r\leq 500$.

\paragraph{Heath networks}
For each of these options, we create $500$ networks where we add $r=1$, 2, 3, 4, 5, 10, 15, 20, 30, 40, 50 reticulations locally, and similarly for less local reticulations. This results in $3\cdot 500 \cdot 11 \cdot 2 = 33000$ networks. For parameter sets used, see Table~\ref{tab:HeathParameters}. 

Note that we do not generate networks with 100, 200, and 500 reticulations with this generator. This is because the Heath network generator uses m-type reticulation events, which add a leaf to the network. Hence, to generate a network with $r$ reticulations, the network must contain at least $r+2$ leaves, so it is impossible to generate networks with 100 leaves and 100 reticulations.

\section{Results}
\label{S3}

\subsection{Topological profiles}

\subsubsection{Balance}



It is displayed in Figure~\ref{fig:Balance} that the balance of the generated networks depends primarily on the balance of the underlying tree model. For the generators based on the Yule model, that is, the LGT network generator and the ZODS network generator, the $B_2$ balance of networks with few reticulations are similar to the beta-splitting networks at $\beta=0$ (Yule model). 



Similarly, for beta-splitting networks, the values of the parameter $\beta$ determining the balance of the underlying trees also determine the $B_2$ balance of the network when the number of reticulations is small. As the number of reticulations increases, the balance differs with respect to the method of adding reticulations. Adding reticulations horizontally at leaf nodes makes the network more balanced, while adding reticulations locally leads to less balanced networks. 

Regarding the MCMC network sampler, networks with more reticulations are more balanced on average than networks with fewer reticulations. This is because, in the MCMC network sampler, the networks with few reticulations are relatively unbalanced compared to the other network generators, as unbalanced networks are over-represented among all distinguishable arrangements. 
In fact, when samples are taken after sufficiently many steps, the MCMC network sampler theoretically behaves exactly the same as the PDA model (beta-splitting model with $\beta=-1.5$) on trees, which is known to be relatively unbalanced.






\begin{figure}[h!]
    \centering
    \includegraphics[scale = 0.35]{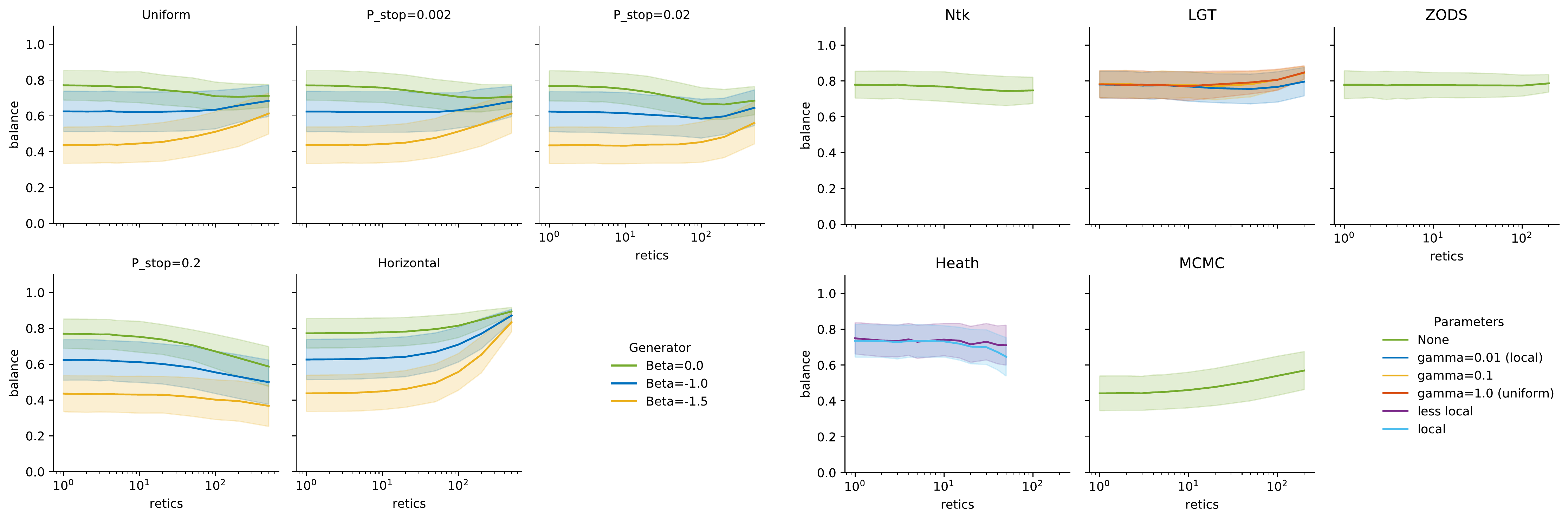}
\caption{Dependence of the $B_2$ balance on the reticulation number in network generators.}
    \label{fig:Balance}
\end{figure}

    

\subsubsection{Cherries and reticulated cherries}

The relations between the number of cherries, the number of reticulated cherries and the number of reticulations for different network generators are displayed in Figure~\ref{fig:CherriesYule}. For all network generators examined, the number of cherries decreases as the number of reticulations increases. This is because when adding hybrid edges, existing cherries can be removed while new cherries can never be created. It is clear that for network generators based on a branching processes, the initial number of cherries depends on the underlying tree shapes. In particular, the network generators based on the Yule model, namely the LGT and ZODS network generators, have similar initial number of cherries to the beta-splitting network with $\beta=0$ (Yule model). It can also be observed that the initial number of cherries of the NTK network generator is similar to the beta-splitting network with $\beta=0$ as well, and the number of cherries of the MCMC network sample resembles the PDA networks (the beta-splitting network with $\beta=0$).

For the LGT, ZODS, and Heath network generators, the number of reticulated cherries increases with the number of reticulations. The rate with which the number of reticulated cherries increases depends on the reticulation structure used in the process. Specifically, the ZODS network generator uses y-type reticulation structure, and the number of reticulated cherries increases the slowest;
The LGT network generator uses n-type reticulation structure, and the number of reticulated cherries increases faster;
the Heath network generator uses m-type reticulation structure, and the number of reticulated cherries increases the fastest. Also note that adding reticulations locally makes the number of reticulated cherries go up faster in the LGT network generator.

\begin{figure}
    \centering
    \includegraphics[scale = 0.35]{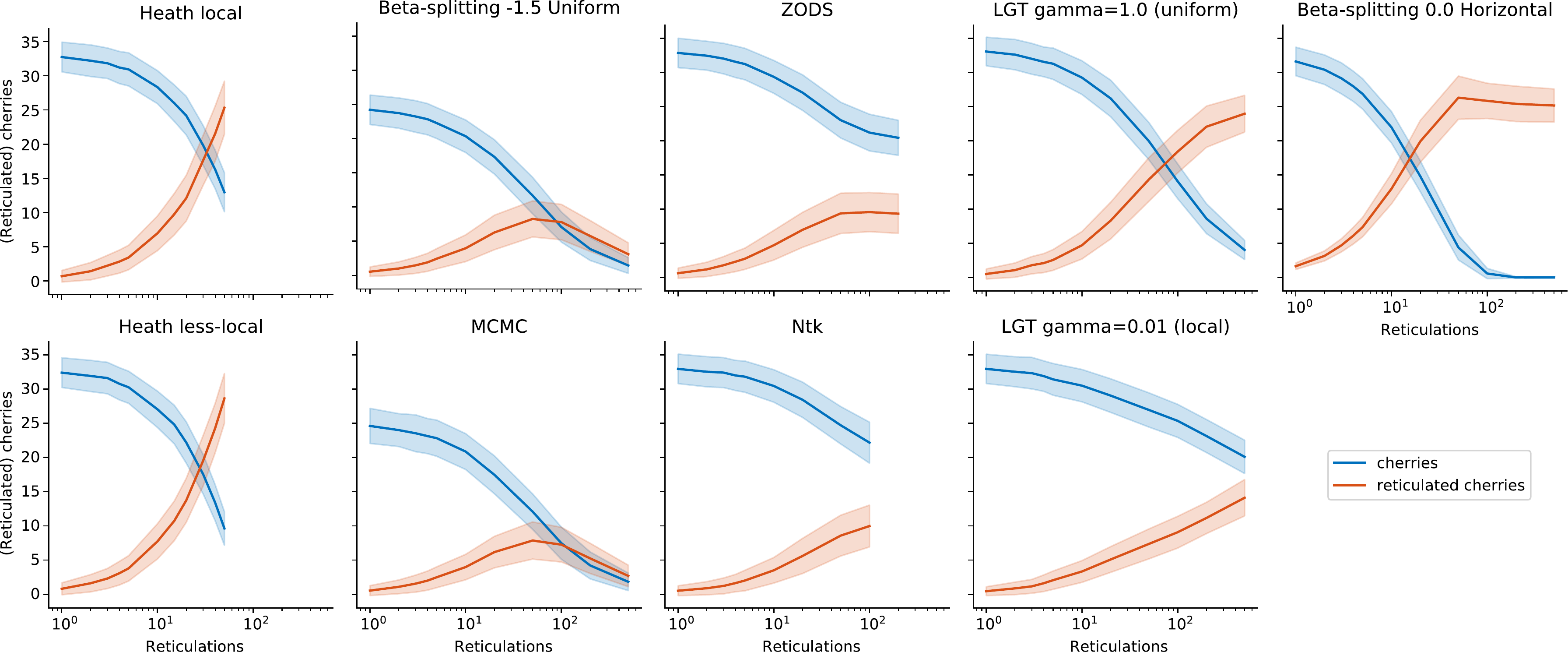}
\caption{Dependence of the number of cherries and reticulated cherries on the number of reticulations for the network generators.}
    \label{fig:CherriesYule}
\end{figure}

\subsubsection{Blobs}

We examine the relations between the number of blobs, the level, the blob size of networks and the number of reticulations in the networks generated by different generators. Figure~\ref{fig:Blobs} displays the results. For all the generators examined, The blob size increases slower than the number of reticulations, but eventually the ratio between the blob size and the number of reticulations becomes constant as the number of blobs in the networks becomes one. More precisely, in a network with $n$ leaf nodes and $r$ reticulations, as $r$ increases and the number of blobs in the networks becomes one, the number of nodes in the largest blobs converges to $n+2r\approx 2r$ and the ratio between the level of the network and the number of reticulations $r$ converges to $1$. 

Trivially, adding edges locally results in more and smaller blobs. The type of locality does affect the profile for the number of blobs subtly as well. The LGT network generator explicitly forces local reticulation events to stay within a blob. Hence, blobs cannot merge in a local event and networks often have multiple blobs even when they have a large number of reticulations. The beta-splitting network generator and the Heath network generator enforce locality differently, where blobs are allowed to merge, which nearly always leads to networks with one blob when a large number of reticulations are added.

Despite the qualitative similarity in the profiles, the Heath network generator with m-type reticulations has distinct profiles with a more significant change in the number of blobs. Similarly, the beta-splitting network generator with different parameter $\beta$ also generates qualitatively similar profiles but with differences in the quantities; see Figure \ref{fig:BlobsBetaSplitByBeta}

\begin{figure}
    \centering
    \includegraphics[scale = 0.35]{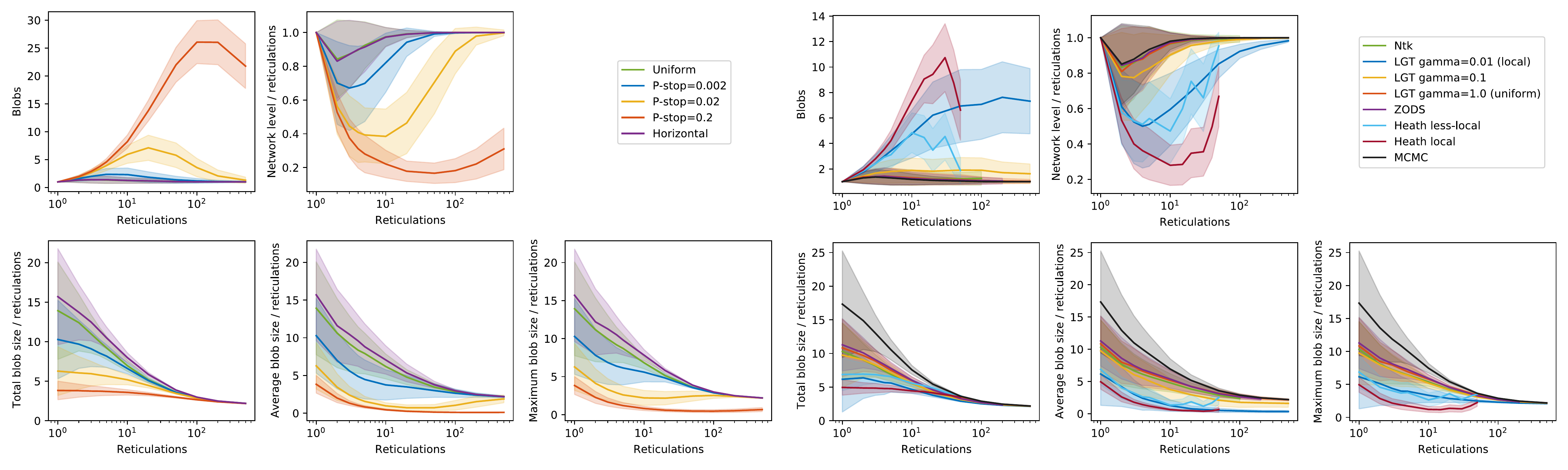}
    \caption{Dependence of the number, size and level of blobs on the reticulation number for network generators.}
    \label{fig:Blobs}
\end{figure}

\subsubsection{Network classes}

Lastly, We study the proportions of generated networks that belong to each of the network classes. Figure \ref{fig:Classes} displays the proportions with respect to the number of reticulations. As the number of reticulations increases, most of the generated networks fall out of any of the four classes. 



For the beta-splitting networks, the beta parameter had negligible influence, so we combined data from all values of betas; see the left panel in Figure~\ref{fig:Classes}. The curves for all different methods of adding reticulations are very similar. One exception is that adding edges horizontally always leads to networks that are orchard, and thus also tree-based. In Lemma~\ref{lem:OnlyHGTOrchard}, we show that it is in fact impossible for this method to produce non-orchard networks. 

For the other network generators, the general results are similar. The exceptions are LGT networks and Heath networks. The LGT network generator, like the beta-splitting model that adds edges horizontally, only produces orchard networks. The Heath network generator is the most restricted, because these only produce tree-child networks as a result of its m-type reticulations.

\begin{figure}[h!]
    \centering
    \includegraphics[scale = 0.42]{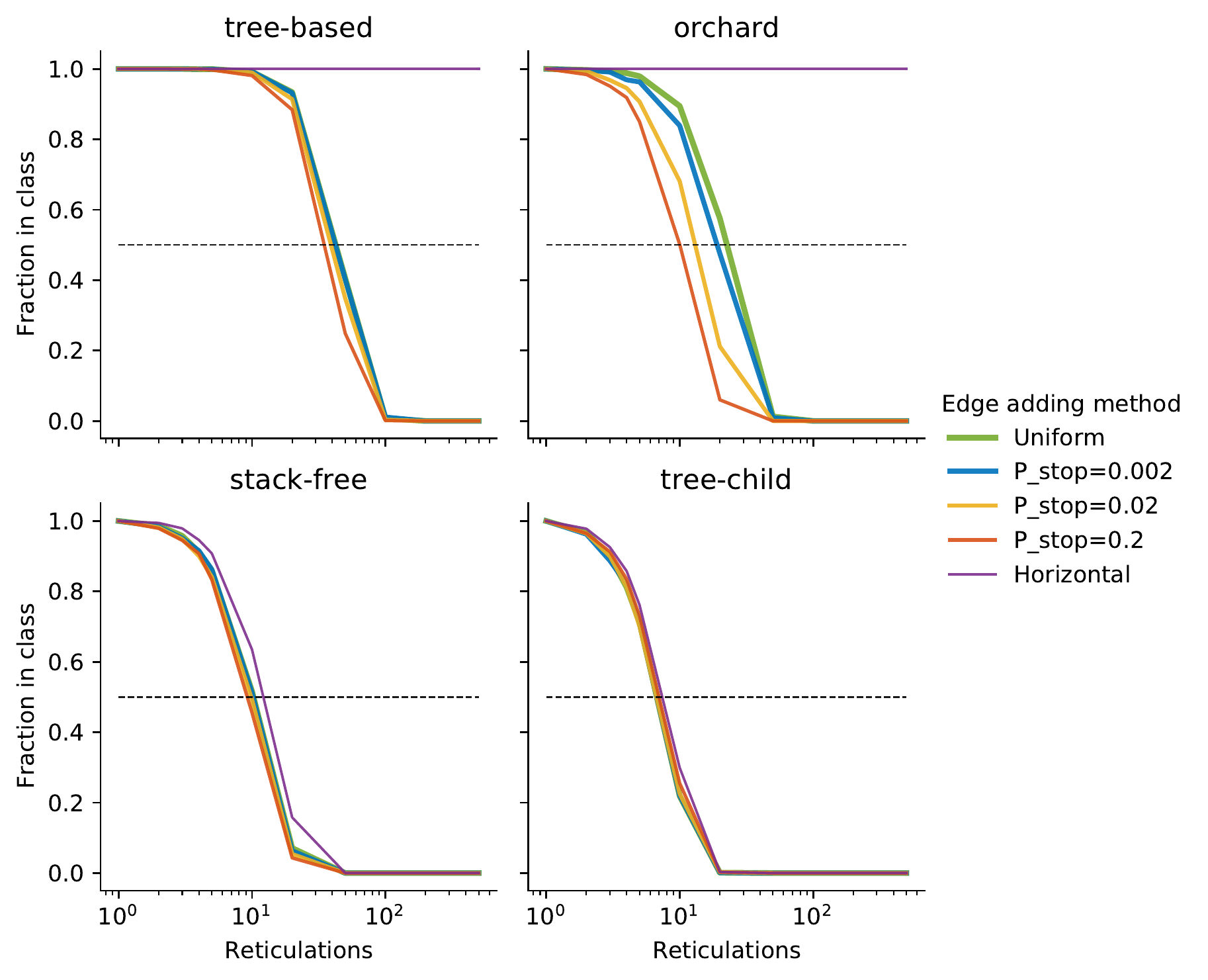}
    \includegraphics[scale = 0.42]{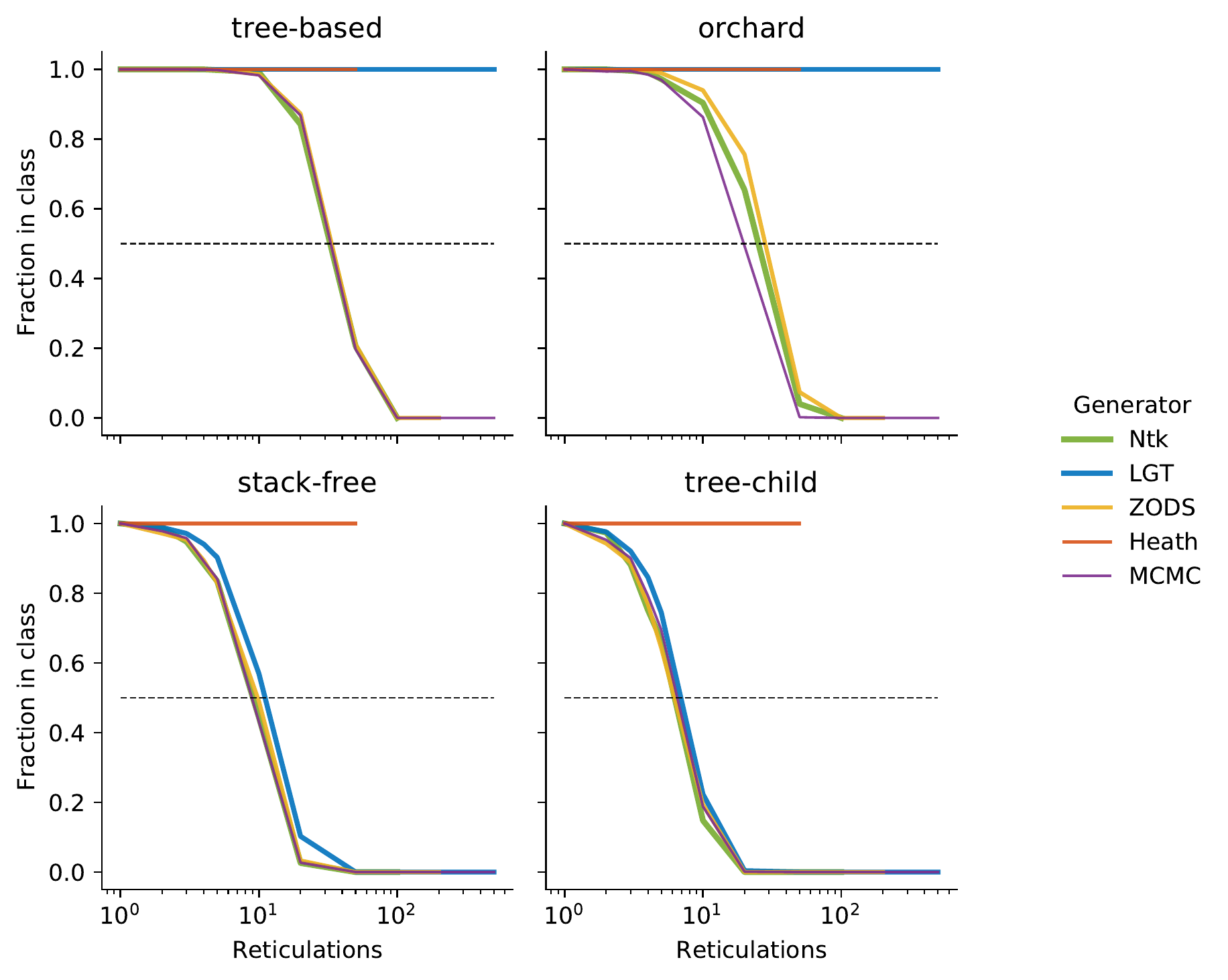}
    \caption{Proportions of generated networks in different network classes for different network generators.}
    \label{fig:Classes}
\end{figure}



\subsection{Polynomial comparison of small networks}

We deploy the polynomial as an alternative comparison between generated networks. We compare the network generators with similar profiles, namely the ZODS network generator versus the NTK network generator and the MCMC network sampler versus PDA network generator (with reticulations added uniformly). As a reference, we also compare the LGT network generator (with reticulations added uniformly, i.e. $\gamma=1.0$) vs the NTK network generator. Note that other than the network classes, the LGT network generator and the NTK network generator have similar topological profiles. Figure~\ref{figpoly} displays the misclassification rates from k-medoids clustering sets of networks, the experiment is performed on networks with 10 leaves, and repeated 50 times for each number of reticulations. A misclassification rate of $0.5$ means that the two sets of generated networks cannot be distinguished, while a misclassification rate of $0$ means the two sets of generated networks can be completely distinguished.
When the number of reticulations in the networks reaches 10, the polynomial distance can distinguish the LGT networks and NTK networks. However, the polynomial distance cannot successfully distinguish the ZODS networks from the NTK netowrks or MCMC networks from the PDA networks.


\begin{figure}
    \centering
    \includegraphics[scale = 0.275]{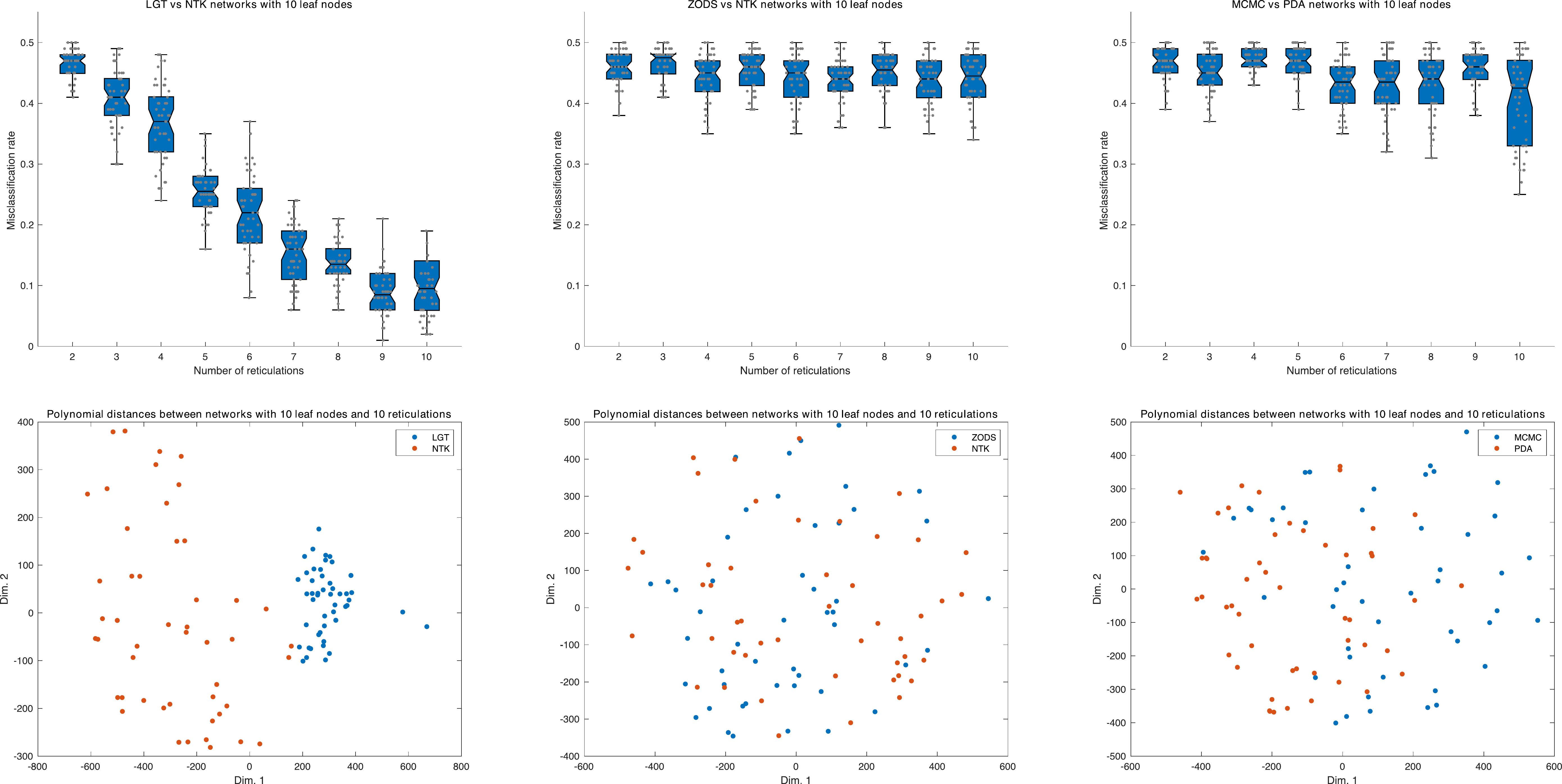}
\caption{Top: the misclassification rates in clustering LGT vs NTK, ZODS vs NTK and MCMC vs PDA networks with 10 leaf nodes by k-medoids. Bottom: the multi-dimensional scaling visualization of polynomial distances between full polynomials of sets of LGT vs NTK, ZODS vs NTK and MCMC vs PDA networks networks with 10 leaf nodes and 10 reticulations.}
    \label{figpoly}
\end{figure}


\section{Discussion}
\label{S4}

We have compared the topological profiles of several network generators. In general, the topological characteristics are related to the underlying tree shapes when the number of reticulations is small. For generators based on a branching process, as the number of reticulations increases, the topological characteristics change with respect to the type of integrated reticulation structures and the locality of added reticulations.
For example, y-type reticulations (ZODS network generator) can be used to create any network in any network class, and they lead to a relatively large number of cherries which seems to persist even for very large numbers of reticulations. This large number of cherries can be explained by the relatively large number of speciation events needed to generate networks with larger numbers of reticulations. Using n-type reticulations (LGT network generator) always leads to orchard networks, and a relatively high balance compared to the other generators. While using m-type reticulations (Heath network generator) always leads to tree-child networks, and to a sharp increase in the number of reticulated cherries, because each m-type reticulation involves two reticulated cherries.

For generators that are not based on a branching process, we have found resemblances between the NTK network generator and the ZODS network generator, as well as between MCMC network sampler and the PDA network generator (with reticulations added uniformly). The former resemblance could be explained by seeing the ZODS network generator as having a predefined set of nodes with degree constraints, which are connected in a top-down fashion. This makes it quite similar to the NTK network generator, but whose description does not explicitly include a top-down order for connecting the nodes. The resemblance between the MCMC network sampler and the PDA network generator is less obvious, as there only a clear relation between these methods when there are no reticulations. It remains an open problem whether the theoretical equivalence for these models on trees can be extended formally to an equivalence on networks.

These resemblances have been further analyzed by comparing small networks generated by the set of generators using the polynomial distance. We found that among the pairs of network generators with similar topological profiles, the polynomial distance can distinguish the LGT networks from the NTK networks, but cannot successfully distinguish ZODS networks from the NTK networks or MCMC networks from PDA networks. We limited our analysis to small networks due to the computation complexity of the polynomial. For larger phylogenetic networks, we can compute the average of coefficients over the set of spanning-tree polynomials to avoid too many multiplications of polynomials. Consider the polynomial coefficients as vectors, we can also compute the diversity \cite{BRYANT2012} of the set of spanning-tree polynomials of a network. This would provide additional information about the topology of a network. 

Perhaps surprisingly, we observed that adding reticulations locally does not only affect the number of blobs and their sizes, but also affects the balance of the networks. This effect is more prominent in the beta-splitting and Heath network generators than in the LGT network generator. It would be interesting to see whether it can be explained mathematically why local reticulations lead to networks with lower balance, and why this effect is not as strong as in the LGT network generator.



As we mentioned in the introduction, there are two main applications for network generators in phylogeny reconstruction: as a tool to validate network reconstruction methods, and as prior for Bayesian methods. The profiles and the analysis of the network topology presented in this paper can help select suitable generators in these tasks. 
It is an important question in validation which generator produces realistic networks that resemble networks reconstructed from data. This question is discussed for phylogenetic trees \cite{blum2006random}, but is not investigated for phylogenetic networks. One of the gaps in this investigation is the lack of realistic networks reconstructed from data. 
Once sufficient realistic networks become available, for example through reconstructions \cite{morando2020phylogenomic,yang2021extensive} or through simulated evolution \cite{cuypers2015endless,van2016evolution,van2020slightly}, the topological profiles and the polynomial can be utilized to compare the natural evolutionary processes with simulated evolution models or the random network generators, and networks reconstructed from data with the generated random networks. 


To provide a more flexible framework for adding hybridization events to tree generators, we also developed a new network generator based on the Heath branching process. This network generator uses the m-type reticulation structure, and hybrid events are more likely for closely related species, which is modelled by a distance that we update efficiently after each event. Consequently, the Heath network generator has different topological profiles from the other network generators. For all these network generators based on a branching process, we did not include extinction events, which would change the topological profiles; see \ref{sec:HeathExt}. Moreover, extinction would allow network generators using m-type reticulations to have more reticulations than leaves, and network generators with m-type and n-type reticulations to generate networks outside the class of tree-based networks. 
 
In the perspective of Bayesian methods, our results suggest that the priors introduced in \cite{wen2016bayesian} and \cite{zhang2018bayesian} can behave quite differently in terms of balance and number of (reticulated) cherries, for in the former an explicit prior is absent so the actual prior is the MCMC network sampler, and in the latter, the prior is the ZODS network generator. Moreover, if a generator that is unable to generate all networks (e.g. the LGT network generator) is used as prior, then there exist networks that cannot be traversed by the random walk in sampling posterior distribution. Using such a prior thus requires new investigations on the connectedness of the spaces of networks that can be generated by such a generator, like for tree-child networks in \cite{bordewich2017lost,klawitter2020spaces}.

\section*{Acknowledgments}

R.J. was supported by the Netherlands Organization for Scientific Research (NWO) Vidi grant 639.072.602. P.L. was supported by the grant of the Federal Government of Canada's Canada 150 Research Chair program to Prof.~C. Colijn.

\bibliographystyle{abbrvnat}
\bibliography{bibliography.bib}

\clearpage

\appendix
\section{Description and pseudo-code for the new network generators}
\subsection{Beta-splitting networks}
Here, we describe how the beta-splitting network generator adds edges to a tree or existing network. Algorithms~\ref{alg:AddEdgesUniform} and~\ref{alg:AddEdgesHorizontal} are quite straightforward: they either choose two edges uniformly at random from the network, or two \emph{leaf}-edges uniformly at random. In Algorithm~\ref{alg:AddEdgeLocal}, which adds the reticulations locally, the two edges are chosen as follows. First, one edge of the network is chosen uniformly at random, and then a random walk is performed through the network (with uniformly chosen steps). This random walk stops after each step with a stopping probability $P_{\text{stop}}$. The higher $P_{\text{stop}}$, the more local the reticulations are. If $P_{\text{stop}}$ is low,this random walk may take a very large number of steps. Hence, in our implementation, we also use a fixed maximal number of steps for this random walk.

\vspace{\baselineskip}
\begin{algorithm}[H]\label{alg:AddEdgesUniform}
 \KwData{A network $N$}
 \KwResult{A network $N'$ obtained from $N$ by adding one arc.}
 Set $N'=N$\;
 Pick two edges $e_1,e_2$ of $N'$ uniformly at random\;
 Subdivide these edges with nodes $v_1$, $v_2$\;
 \uIf{$v_2$ is above $v_1$}{
     Add the arc $(v_2,v_1)$ to $N'$\;
 }
 \Else{
     Add the arc $(v_1,v_2)$ to $N'$\;
 }
\Return $N'$\;
\caption{\textsc{AddEdgeUniform}$(N)$}
\end{algorithm}

\vspace{\baselineskip}
\begin{algorithm}[H]\label{alg:AddEdgeLocal}
 \KwData{A network $N$ and a probability $p_{stop}$.}
 \KwResult{A network $N'$ obtained from $N$ by adding one arc.}
 Set $N'=N$\;
 \While{True}{
    Pick an edges $e_1$ of $N$ uniformly at random\;
    Let $(u,v)$ be a random orientation of $e_1$\;
    Let $\mathrm{stop\_value}=1$\;
    \While{$\mathrm{stop\_value}>p_{stop}$}{
        Set $u=v$\;
        Let $v$ be a random neighbour of $v$\;
        Sample $\mathrm{stop\_value}$ uniformly at random from $[0,1]$\;
    }
    Let $e_2$ be the orientation of $(u,v)$ as found in $N$\;
    \If{$e_1\neq e_2$}{
        Quit the While-loop\;
    } 
 }
Subdivide $e_1$ and $e_2$ with nodes $v_1$, $v_2$\;
\uIf{$v_2$ is above $v_1$}{
    Add the arc $(v_2,v_1)$ to $N$\;
}
\Else{
    Add the arc $(v_1,v_2)$ to $N$\;
}
\Return $N$\;
\caption{\textsc{AddEdgeLocal}$(T,p_{stop})$}
\end{algorithm}

\vspace{\baselineskip}
\begin{algorithm}[H]\label{alg:AddEdgesHorizontal}
 \KwData{A network $N$}
 \KwResult{A network $N'$ obtained from $N$ by adding one arc.}
 Set $N'=N$\;
 Pick two leaves $l_1,l_2$ of $N'$ uniformly at random\;
 Subdivide the incoming edges of these leaves with nodes $v_1$, $v_2$\;
 Add the arc $(v_1,v_2)$ to $N'$\;
\Return $N'$\;
\caption{\textsc{AddEdgeHorizontal}$(N)$}
\end{algorithm}
\vspace{\baselineskip}

\subsection{Heath networks}
In this section, we give more details pertaining the implementation of our network generator based on the Heath branching process \cite{heath2008taxon}. Algorithms~\ref{alg:UpdateRate} and~\ref{alg:Heath} contain the pseudo-code for the model, and in the next subsection, we show that updating the distances as we do in Algorithm~\ref{alg:Heath} is correct.

\vspace{\baselineskip}
\begin{algorithm}[H]\label{alg:UpdateRate}
 \KwData{A rate $r$, parameters $r_{\alpha},r_{\beta}$ for the prior gamma distribution of $r$, and parameters $u_{\alpha},u_{\beta}$ for the update gamma distribution.}
 \KwResult{A new rate $r'$}
 Sample $u$ from $\Gamma(u_{\alpha},u_{\beta})$\;
 Let $r'=ur$ be the proposed new rate\;
 With probability $\min(1,\pdf_{\Gamma(r_{\alpha},r_{\beta})}(r')/\pdf_{\Gamma(r_{\alpha},r_{\beta})}(r'))$, \Return $r'$\;
 Otherwise, \Return $r$\;

\caption{\textsc{UpdateRate}$(r,r_{\alpha},r_{\beta},u_{\alpha},u_{\beta})$}
\end{algorithm}

\subsubsection{Updating distances}\label{sec:HeathUpdateDistance}

As an example of an extension of a Markov model to a more flexible birth-hybridization model, we extend the model presented in \cite{heath2008taxon}, which uses auto-correlated speciation and extinction rates. To create networks, we add the options of including hybridization events. These events can be added to other Markov splitting models in the same way. So, if different or more realistic splitting models are desired (e.g., with speciation bursts) hybridization can easily be added to those as well.

Hybridization in this model should be thought of as hybrid speciation, and not hybrid introgression, that is, m-type reticulation model are used. Indeed, a hybridization event in this generator takes two extant taxa and combines them into a new taxon. The two parental taxa stay present in the network, so the number of taxa in the network increases by one.

In our model, the hybridization rate between two taxa depends on their distance. To keep it simple, we use the piece-wise linear function 
\[
f_{hyb}(d) = \begin{cases} h_{lr}                                           &\mbox{if } d \leq h_l \\
                           h_{lr}+\frac{(h_{rr}-h_{lr})(d-h_l)}{h_r-h_l}    &\mbox{if } h_l<d<h_r  \\
                           h_{rr}                                           &\mbox{if } d \geq h_r \end{cases} .\]

Recomputing the total rates for speciation $\spe(N)=\sum_{x}\spe(x)$ and extinction $\ext(N)=\sum_{x}\ext(x)$ is quite simple, as it only requires one summation over the individual speciation and extinction rates $\spe(x)$ and $\ext(x)$ of all the extant taxa $x$. Fully recomputing $\hyb(N)$ would entail finding all up-down paths between each pair of leaves. As there can be exponentially many up-down paths between a given pair of leaves, this full recomputation is not feasible. 

Hence, to efficiently compute $\hyb(N)$ in our implementation, we keep track of the weighted distances between all pairs of leaves. Updating these distances requires few calculations for each of the types of events, as we will show next.

\vspace{\baselineskip}
\begin{algorithm}[H]\label{alg:Heath}
Let $N$ be the tree on one leaf $l$\;
Sample speciation rate $\spe(l)=\spe(N)$ of $l$ from $\Gamma(s_{\alpha},s_{\beta})$\;
Sample extinction rate $\ext(l)=\ext(N)$ of $l$ from $\Gamma(s_{\alpha},s_{\beta})$\;
Sample hgt rate $\hgt(l)=\hgt(N)$ of $l$ from $\Gamma(s_{\alpha},s_{\beta})$\;
Set $\hyb(N)=0$\;
Sample the next event time $t^+=$ from an exponential distribution with parameter $1/(\spe(l)+\ext(l)+\hgt(l))$\;
\While{$|\Ta(N)| < $ max\_taxa and $|\Ta_{extant}(N)|>0$}{
    Pick a random event by their weights $\spe(N),\ext(N),\hgt(N),$ and $\hyb(N)$\;
    \If{speciation}{
        Pick a random extant leaf $l$, weighted by by their speciation rates\;
        Attach two new edges with length $0$ to this leaf\;
        Update the rates of l to get rates for the new leaves\;
        Replace $l$ with the two new leaves in the set of extant leaves\;
    }
    \ElseIf{extinction}{
        Pick a random extant leaf $l$, weighted by by their extinction rates\;
        Remove $l$ from the set of extant leaves\;
    }
    \ElseIf{hgt}{
        Pick a random extant leaf $l$, weighted by by their hgt rates\;
        Pick another leaf $m$ uniformly at random\;
        Add length-$0$ edges $(m,l)$, $(l,l')$, and $(m,m')$, where $l'$ and $m'$ are new leaves\;
        Pick an inheritance probability $p$ along $(m,l)$ uniformly at random from $[0,hgt_{inh}]$\;
        Copy the rates of $m$ to $m'$\;
        Set the rates of $l'$ to be $p r(m)+(1-p) r(l)$, and then update these\;
        Replace $l$ and $m$ with $l'$ and $m'$ in the set of extant taxa\;
    }
    \ElseIf{hybridization}{
        Pick a random pair of extant leaves $(l,m)$, weighted by their hybridization rates $f_{hyb}(d(l,m))$\;
        Add length-$0$ edges $(m,x)$, $(l,x)$, $(l,l')$, $(m,m')$, and $(x,q')$, where $l'$, $q'$, and $m'$ are new leaves\;
        Pick an inheritance probability $p$ along $(m,x)$ uniformly at random from $[0,1]$\;
        Copy the rates of $m$ and $l$ to $m'$ and $l'$\;
        Set the rates of $q'$ to be $p r(m)+(1-p) r(l)$, and then update these\;
        Replace $l$ and $m$ with $l'$, $q'$, and $m'$ in the set of extant taxa\;
    }
    Extend the length of all pendant arcs leading to extant taxa by $t^+$\;
    Recompute $\spe(N)$, $\ext(N)$, $\hgt(N)$, and $\hyb(N)$\;
    Sample $t^+=$ from an exponential distribution with parameter $1/(\spe(N)+\ext(N)+\hgt(N)+\hyb(N))$\;
}
\Return $N$\;
\caption{\textsc{Heath}$(s_{\alpha},s_{\beta},e_{\alpha},e_{\beta},hgt_{\alpha},hgt_{\beta},hgt_{inh},u_{\alpha},u_{\beta},h_l,h_r,h_{lr},h_{rr})$}
\end{algorithm}

We will now show how to efficiently update the distances between all pairs of leaves efficiently after each of the following events: speciation, extinction, horizontal gene transfer/hybrid introgression, and hybrid speciation. In all these events, we assume all new edges have length 0 and are extended only later.

\paragraph{Speciation}
Suppose the speciation event splits the extant taxon $l$ into two new taxa $m$ and $m'$. As the length of both the two new edges is zero, the distance between $m$ and $m'$ is 0. For any other leaf $q$, we get $d(m,q)=d(m',q)=d(l,q)$; and for each pair of leaves not involving $m$ or $m'$, the distance is not changed by the speciation event.

\paragraph{Extinction}
An extinction event does not affect the distances between pairs of extant taxa. The only change is that the set of extant taxa changes, i.e., in the set of pairs that we compute the hybridization rate for.

\paragraph{HGT}
Suppose the HGT event uses $l$ as accepting taxon, and $m$ as donating taxon. Adding the arc $(m,l)$ only affects the distances between $l$ and other leaves. Indeed, for pairs of leaves where $l$ is not involved, the set of up-down paths is not affected by adding $(m,l)$. If the inheritance probability along $(m,l)$ is $p$, then the new distances involving $l$ are $d(l,q)=pd(m,q)+(1-p)d(l,q)$.

\paragraph{Hybridization}
Suppose this event adds a hybrid $q$ between $l$ and $m$. The new arcs $(l,q)$ and $(m,q)$ can only be part of up-down paths between $q$ and another leaf. Hence, the only affected distances involve $q$ and another leaf. If the inheritance probability along $(m,q)$ is $p$, then the new distances involving $q$ are $d(q,x)=pd(m,x)+(1-p)d(l,x)$.

\paragraph{Merging}
Suppose this event merges $l$ and $m$ into a reticulation $r$ and introduces the new edge $(r,q)$ (of length 0). Note that the only affected distances involve $q$ and another leaf. Let $p_l$ and $p_m$ be the parents of $l$ and $m$, and let the inheritance probability along $(p_l,r)$ be $p$, then the new distances involving $q$ are $d(q,x)=pd(l,x)+(1-p)d(m,x)$.

\paragraph{Extending pendant arcs}
By extending the pendant arcs of the extant taxa by $t^+$, each up-down path between a pair of extant taxa becomes $2t^+$ longer. Hence, the weighted sums of the distances along these up-down paths also become $2t^+$ longer.


\section{Network generators and classes}\label{sec:GeneratorClassesProofs}
In the following lemmas, we prove that certain processes inherently lead to certain classes of networks, and other are able to produce any network.

\begin{lemma}\label{lem:ZODSAllNetworks}
Let $N$ be an arbitrary network topology, then using merging (adding edges y-type) and speciation can generate $N$.
\end{lemma}
\begin{proof}
We prove this lemma by induction on the number of reticulations. First, for the basis of the induction, note that each network with 0 reticulations (i.e., each tree) can be generated by only using speciation events. 

Now suppose that each network with at most $r$ reticulations can be generated, and let $N$ be an arbitrary network with $n$ leaves and $r+1$ reticulations. If $N$ has a leaf that is directly below a reticulation node. Removing this leaf, and splitting its parent into two nodes, one for each incoming arc, we get a network $N'$ with $n+1$ leaves and $r$ reticulations from which $N$ can be obtained by one merging event. By the induction hypothesis, $N'$ can be generated and, therefore, so can $N$.

If $N$ contains no leaf directly below a reticulation, then there must be a reticulation in $N$ with a pendant subtree directly below it. Let $N'$ be the network obtained from $N$ by removing this subtree. It is clear that $N$ can be obtained from $N'$ using only speciation events, as each tree can be generated using speciation events. Moreover, by the arguments in the previous paragraph, $N'$ can be generated as it has $r+1$ reticulations and a leaf directly below a reticulation. We conclude that any network with $r+1$ reticulations can be produced using merging and speciation events. By induction, this implies that all networks can be generated thusly.
\end{proof}

\begin{lemma}\label{lem:OnlyHGTOrchard}
Using HGT (adding reticulations n-type at the leaves) together with speciation generates only orchard networks. 
\end{lemma}
\begin{proof}
Note that speciation and adding reticulations using the n-type are exactly the methods to ``add a pair to a network'' in the sense of \cite{janssen2021cherry}. The lemma now follows immediately, because these additions of pairs can be reversed, which gives a cherry-picking sequence for the network.
\end{proof}

\begin{lemma}\label{lem:OnlyHybTC}
Using hybridization (adding reticulations m-type at the leaves) together with speciation generates only temporal tree-child networks, and each (binary) temporal tree-child network can be generated this way.
\end{lemma}
\begin{proof}
The first part follows because we can label each internal node with the number of the step (event) it was introduced in. This gives a temporal labelling of the network. Moreover, the such a network must be tree-child because an m-type reticulation event can never lead to two neighbouring reticulations or a tree node with two reticulation children. For the second part, note that the temporal labelling can be modified so that for each time $t$, the internal nodes with label $t$ either represent a speciation event, or an m-type hybridization event.
\end{proof}

\begin{lemma}\label{lem:ExtinctionAllNetworks}
HGT (n-type) or hybridization (m-type) together with extinction and speciation can generate any network.
\end{lemma}
\begin{proof}
This follows directly from Lemma~\ref{lem:ZODSAllNetworks}, because each merging event can be substituted for a HGT event and one extinction event, or a hybridization event and two extinction events.
\end{proof}

\section{Supplementary figures and tables}
\setcounter{figure}{0}
\setcounter{table}{0}
In this section, we provide additional data in the form of tables and figures. In Table~\ref{tab:Classes}, for each generator, we give an estimate for the number of reticulations at which half of the generated networks are in each network class. This value is found by finding the intersection of the horizontal line where the fraction is a half with the lines in Figure~\ref{fig:Classes}. In Table~\ref{tab:Balance}, we show the average balance for each of the generators and each number of reticulations. The stars show that adding reticulations significantly impacts the balance of the networks compared to the networks with 1 reticulation for the same network generator. Figure~\ref{fig:CherriesBeta} shows the profiles for the number of (reticulated) cherries in the beta-splitting model. Note that the initial number of cherries depends primarily on the beta parameter, while all profiles, except when edges are added horizontally, are qualitatively the same. Figure~\ref{fig:BlobsBetaSplitByBeta} shows the dependence of the number of blobs and their sizes for the beta-splitting model. Note that the profiles are qualitatively the same, although the exact values differ slightly, as a result of the underlying tree shape.

\begin{table}[]
\caption{The number of reticulations at which half the networks are in a class for the different generators and parameters. Obtained by interpolation from the data with log-transformed number of reticulations, because in those data, the curves look more or less sigmoid (Figure~\ref{fig:Classes}), and so that we do not underestimate the effects at small numbers of reticulations.}\label{tab:Classes}
\vspace{\baselineskip}
\tiny
\centering
\begin{tabular}{lll|c|c|c|c|}
\cline{4-7}
                                    &                                &                & \multicolumn{4}{c|}{Class}                                                                                                                                                                            \\ \hline
\multicolumn{1}{|l|}{Generator}     & \multicolumn{1}{l|}{Parameter} & Locality       & \multicolumn{1}{l|}{Tree-based}                 & \multicolumn{1}{l|}{Orchard}                    & \multicolumn{1}{l|}{Stack-free}                 & \multicolumn{1}{l|}{Tree-child}                 \\ \hline
\multicolumn{1}{|l}{Beta-splitting} & Beta=0.0                       & Uniformly      & 41.2                                            & 24.0                                            & 10.1                                            & 6.6                                             \\ \cline{4-7} 
\multicolumn{1}{|l}{}               &                                & P\_stop=0.002  & 41.6                                            & 21.4                                            & 11.0                                            & 7.0                                             \\ \cline{4-7} 
\multicolumn{1}{|l}{}               &                                & P\_stop=0.2    & 33.9                                            & 9.4                                             & 8.9                                             & 6.9                                             \\ \cline{4-7} 
\multicolumn{1}{|l}{}               &                                & Horizontal     & \cellcolor[HTML]{000000}                        & \cellcolor[HTML]{000000}{\color[HTML]{000000} } & 12.6                                            & 7.6                                             \\ \hline
\multicolumn{1}{|l}{Zhang Ntk}      &                                &                & 32.7                                            & 25.2                                            & 9.1                                             & 6.2                                             \\ \hline
\multicolumn{1}{|l}{LGT Generator}  &                                & Locality\_0.01 & \cellcolor[HTML]{000000}                        & \cellcolor[HTML]{000000}                        & 10.5                                            & 6.9                                             \\ \cline{4-7} 
\multicolumn{1}{|l}{}               &                                & Locality\_0.1  & \cellcolor[HTML]{000000}                        & \cellcolor[HTML]{000000}                        & 10.4                                            & 6.6                                             \\ \cline{4-7} 
\multicolumn{1}{|l}{}               &                                & Locality\_1.0  & \cellcolor[HTML]{000000}                        & \cellcolor[HTML]{000000}                        & 12.5                                            & 7.4                                             \\ \hline
\multicolumn{1}{|l}{ZODS}           &                                &                & 33.5                                            & 28.2                                            & 9.8                                             & 6.2                                             \\ \hline
\multicolumn{1}{|l}{Heath}          & keep extinct                   & less local     & \cellcolor[HTML]{000000}                        & \cellcolor[HTML]{000000}                        & \cellcolor[HTML]{000000}                        & \cellcolor[HTML]{000000}                        \\ \cline{4-7} 
\multicolumn{1}{|l}{}               &                                & local          & \cellcolor[HTML]{000000}                        & \cellcolor[HTML]{000000}                        & \cellcolor[HTML]{000000}                        & \cellcolor[HTML]{000000}                        \\ \cline{4-7} 
\multicolumn{1}{|l}{}               & remove extinct                 & less local     & ?                                               & ?                                               & 26.8                                            & 13.9                                            \\ \cline{4-7} 
\multicolumn{1}{|l}{}               &                                & local          & ?                                               & ?                                               & 24.6                                            & 13.9                                            \\ \cline{4-7} 
\multicolumn{1}{|l}{}               & no extinction                  & less local     & \cellcolor[HTML]{000000}{\color[HTML]{000000} } & \cellcolor[HTML]{000000}{\color[HTML]{000000} } & \cellcolor[HTML]{000000}{\color[HTML]{000000} } & \cellcolor[HTML]{000000}{\color[HTML]{000000} } \\ \cline{4-7} 
\multicolumn{1}{|l}{}               &                                & local          & \cellcolor[HTML]{000000}{\color[HTML]{000000} } & \cellcolor[HTML]{000000}{\color[HTML]{000000} } & \cellcolor[HTML]{000000}{\color[HTML]{000000} } & \cellcolor[HTML]{000000}{\color[HTML]{000000} } \\ \hline
\multicolumn{1}{|l}{MCMC}           &                                &                & 33.1                                            & 19.7                                            & 8.9                                             & 6.5                                             \\ \hline
\end{tabular}
\end{table}


\begin{table}[]
\caption{The average balance of networks for different network generators. Stars indicate whether the balance is significantly different from the networks with 1 reticulation of the same generator ($*: p<0.01$, $**: p<0.001$). }\label{tab:Balance}
\vspace{\baselineskip}
\tiny
\centering
\begin{tabular}{ll|lllllllll|}
\cline{3-11}
                                                                      &               & \multicolumn{9}{c|}{Reticulations}                                                                                                                                                                                                                                           \\ \hline
\multicolumn{1}{|l}{Generator}                                        & Parameters    & \multicolumn{1}{c}{1}      & \multicolumn{1}{c}{2}      & \multicolumn{1}{c}{5}        & \multicolumn{1}{c}{10}       & \multicolumn{1}{c}{20}       & \multicolumn{1}{c}{50}       & \multicolumn{1}{c}{100}      & \multicolumn{1}{c}{200}      & \multicolumn{1}{c|}{500} \\ \hline
\multicolumn{1}{|l}{Beta-splitting beta=0.0}         & Uniform       & \multicolumn{1}{l|}{0.771} & \multicolumn{1}{l|}{0.769} & \multicolumn{1}{l|}{0.762}   & \multicolumn{1}{l|}{0.760}   & \multicolumn{1}{l|}{0.745**} & \multicolumn{1}{l|}{0.730**} & \multicolumn{1}{l|}{0.710**} & \multicolumn{1}{l|}{0.706**} & 0.712**                  \\ \cline{3-11} 
\multicolumn{1}{|l}{}                                                 & P\_stop=0.002 & \multicolumn{1}{l|}{0.770} & \multicolumn{1}{l|}{0.769} & \multicolumn{1}{l|}{0.763}   & \multicolumn{1}{l|}{0.757}   & \multicolumn{1}{l|}{0.744**} & \multicolumn{1}{l|}{0.722**} & \multicolumn{1}{l|}{0.707**} & \multicolumn{1}{l|}{0.699**} & 0.707**                  \\ \cline{3-11} 
\multicolumn{1}{|l}{}                                                 & P\_stop=0.2   & \multicolumn{1}{l|}{0.770} & \multicolumn{1}{l|}{0.768} & \multicolumn{1}{l|}{0.761}   & \multicolumn{1}{l|}{0.752*}  & \multicolumn{1}{l|}{0.737**} & \multicolumn{1}{l|}{0.705**} & \multicolumn{1}{l|}{0.671**} & \multicolumn{1}{l|}{0.635**} & 0.587**                  \\ \cline{3-11} 
\multicolumn{1}{|l}{}                                                 & Horizontal    & \multicolumn{1}{l|}{0.772} & \multicolumn{1}{l|}{0.773} & \multicolumn{1}{l|}{0.774}   & \multicolumn{1}{l|}{0.778}   & \multicolumn{1}{l|}{0.782}   & \multicolumn{1}{l|}{0.796**} & \multicolumn{1}{l|}{0.815**} & \multicolumn{1}{l|}{0.848**} & 0.894**                  \\ \hline
\multicolumn{1}{|l}{Zhang\_Ntk}                                       &               & \multicolumn{1}{l|}{0.779} & \multicolumn{1}{l|}{0.778} & \multicolumn{1}{l|}{0.773}   & \multicolumn{1}{l|}{0.769}   & \multicolumn{1}{l|}{0.756**} & \multicolumn{1}{l|}{0.744**} & \multicolumn{1}{l|}{0.747**} & \multicolumn{1}{l|}{}        &                          \\ \hline
\multicolumn{1}{|l}{LGT\_Generator}                  & Beta=0.01     & \multicolumn{1}{l|}{0.780} & \multicolumn{1}{l|}{0.779} & \multicolumn{1}{l|}{0.777}   & \multicolumn{1}{l|}{0.769}   & \multicolumn{1}{l|}{0.760**} & \multicolumn{1}{l|}{0.755**} & \multicolumn{1}{l|}{0.767*}  & \multicolumn{1}{l|}{0.796*}  & 0.848**                  \\ \cline{3-11} 
\multicolumn{1}{|l}{}                                                 & Beta=0.1      & \multicolumn{1}{l|}{0.782} & \multicolumn{1}{l|}{0.784} & \multicolumn{1}{l|}{0.779}   & \multicolumn{1}{l|}{0.778}   & \multicolumn{1}{l|}{0.771}   & \multicolumn{1}{l|}{0.782}   & \multicolumn{1}{l|}{0.807**} & \multicolumn{1}{l|}{0.845**} & 0.876**                  \\ \cline{3-11} 
\multicolumn{1}{|l}{}                                                 & Beta=1.0      & \multicolumn{1}{l|}{0.781} & \multicolumn{1}{l|}{0.779} & \multicolumn{1}{l|}{0.777}   & \multicolumn{1}{l|}{0.771}   & \multicolumn{1}{l|}{0.781}   & \multicolumn{1}{l|}{0.792}   & \multicolumn{1}{l|}{0.806**} & \multicolumn{1}{l|}{0.847**} & 0.884**                  \\ \hline
\multicolumn{1}{|l}{ZODS}                                             &               & \multicolumn{1}{l|}{0.779} & \multicolumn{1}{l|}{0.779} & \multicolumn{1}{l|}{0.776}   & \multicolumn{1}{l|}{0.778}   & \multicolumn{1}{l|}{0.776}   & \multicolumn{1}{l|}{0.775}   & \multicolumn{1}{l|}{0.774}   & \multicolumn{1}{l|}{0.787}   &                          \\ \hline
\multicolumn{1}{|l}{Heath+extinction}                & less local    & \multicolumn{1}{l|}{0.506} & \multicolumn{1}{l|}{0.512} & \multicolumn{1}{l|}{0.521}   & \multicolumn{1}{l|}{0.505}   & \multicolumn{1}{l|}{0.526}   & \multicolumn{1}{l|}{0.547**} & \multicolumn{1}{l|}{}        & \multicolumn{1}{l|}{}        &                          \\ \cline{3-11} 
\multicolumn{1}{|l}{}                                                 & local         & \multicolumn{1}{l|}{0.506} & \multicolumn{1}{l|}{0.510} & \multicolumn{1}{l|}{0.521}   & \multicolumn{1}{l|}{0.508}   & \multicolumn{1}{l|}{0.525}   & \multicolumn{1}{l|}{0.526}   & \multicolumn{1}{l|}{}        & \multicolumn{1}{l|}{}        &                          \\ \cline{3-11} 
\multicolumn{1}{|l}{Heath+extinction+remove extinct} & less local    & \multicolumn{1}{l|}{0.732} & \multicolumn{1}{l|}{0.727} & \multicolumn{1}{l|}{0.713*}  & \multicolumn{1}{l|}{0.736}   & \multicolumn{1}{l|}{0.700**} & \multicolumn{1}{l|}{0.693**} & \multicolumn{1}{l|}{}        & \multicolumn{1}{l|}{}        &                          \\ \cline{3-11} 
\multicolumn{1}{|l}{}                                                 & local         & \multicolumn{1}{l|}{0.722} & \multicolumn{1}{l|}{0.722} & \multicolumn{1}{l|}{0.716}   & \multicolumn{1}{l|}{0.715}   & \multicolumn{1}{l|}{0.687**} & \multicolumn{1}{l|}{0.639**} & \multicolumn{1}{l|}{}        & \multicolumn{1}{l|}{}        &                          \\ \cline{3-11} 
\multicolumn{1}{|l}{Heath}                           & less local    & \multicolumn{1}{l|}{0.749} & \multicolumn{1}{l|}{0.737} & \multicolumn{1}{l|}{0.730**} & \multicolumn{1}{l|}{0.742}   & \multicolumn{1}{l|}{0.716**} & \multicolumn{1}{l|}{0.711**} & \multicolumn{1}{l|}{}        & \multicolumn{1}{l|}{}        &                          \\ \cline{3-11} 
\multicolumn{1}{|l}{}                                                 & local         & \multicolumn{1}{l|}{0.736} & \multicolumn{1}{l|}{0.734} & \multicolumn{1}{l|}{0.736}   & \multicolumn{1}{l|}{0.732}   & \multicolumn{1}{l|}{0.703**} & \multicolumn{1}{l|}{0.646**} & \multicolumn{1}{l|}{}        & \multicolumn{1}{l|}{}        &                          \\ \hline
\multicolumn{1}{|l}{MCMC}                                             &               & \multicolumn{1}{l|}{0.442} & \multicolumn{1}{l|}{0.444} & \multicolumn{1}{l|}{0.449}   & \multicolumn{1}{l|}{0.461**} & \multicolumn{1}{l|}{0.478**} & \multicolumn{1}{l|}{0.510**} & \multicolumn{1}{l|}{0.540**} & \multicolumn{1}{l|}{0.569**} & 0.579**                  \\ \hline
\end{tabular}
\end{table}

\begin{figure}[h!]
    \centering
    \includegraphics[scale = 0.4]{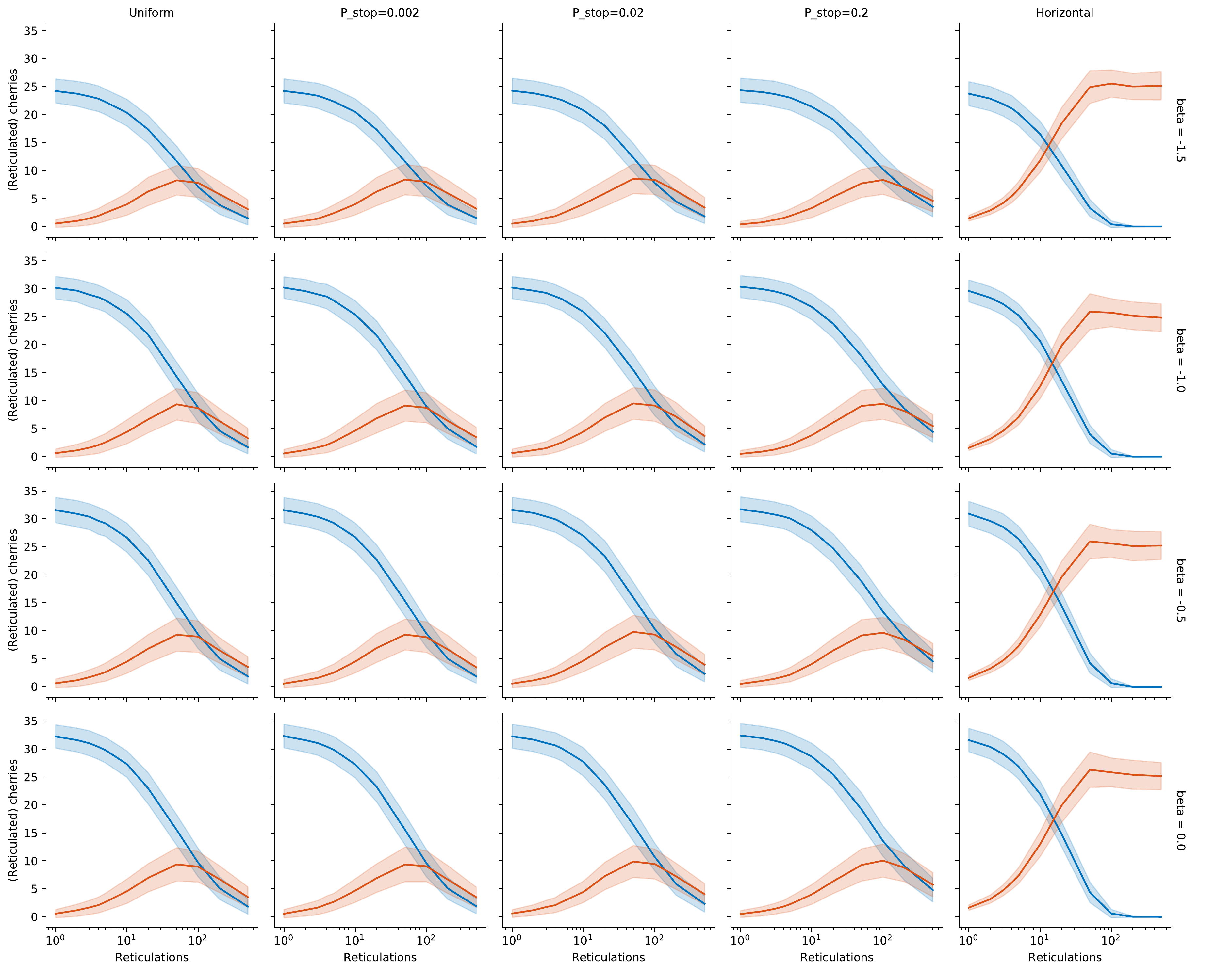}
    \caption{Dependence of the number of cherries (blue) and reticulated cherries (orange) for networks with 100 leaf nodes in the beta-splitting model. }
    \label{fig:CherriesBeta}
\end{figure}



\begin{figure}[h!]
    \centering
    \includegraphics[width = 0.48\textwidth]{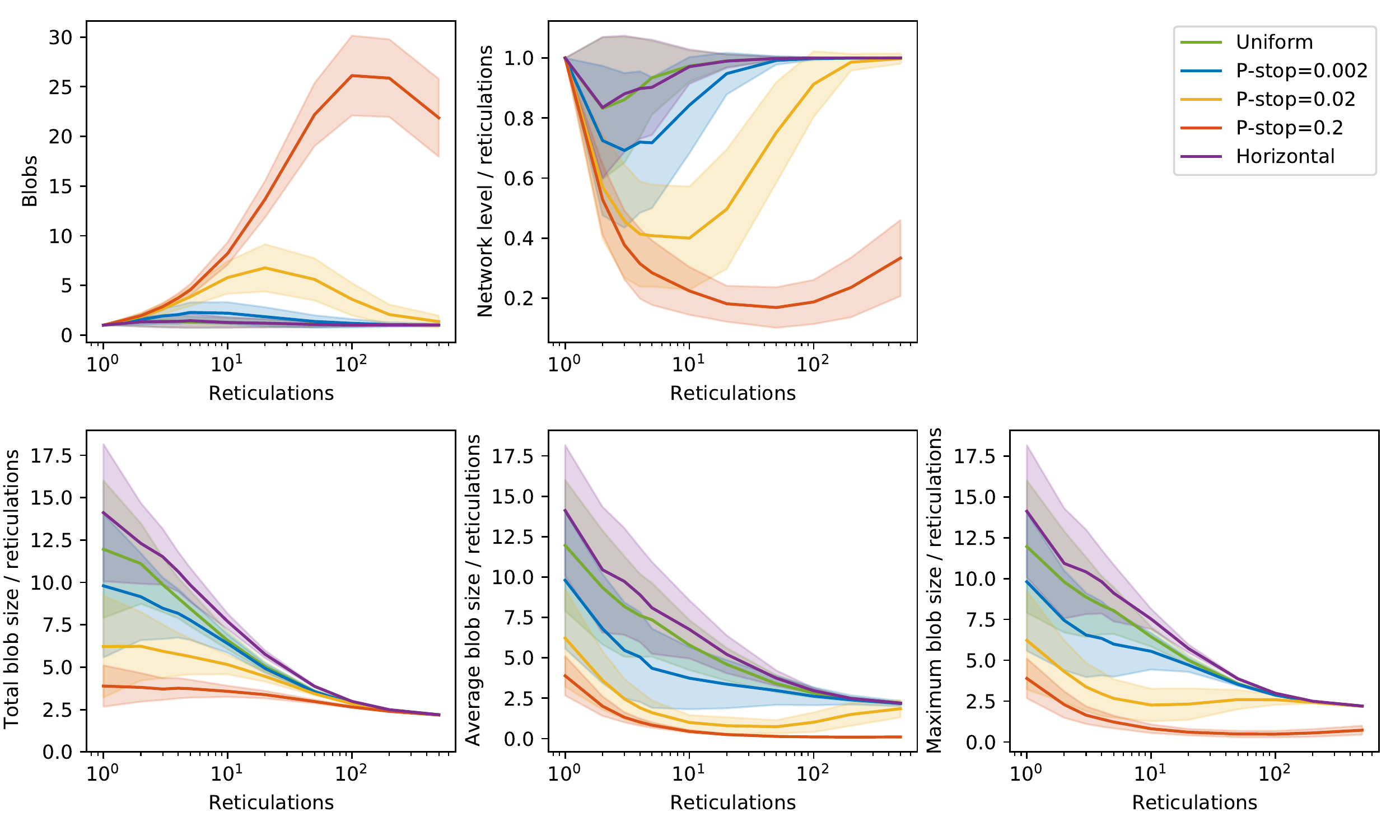}
    \includegraphics[width = 0.48\textwidth]{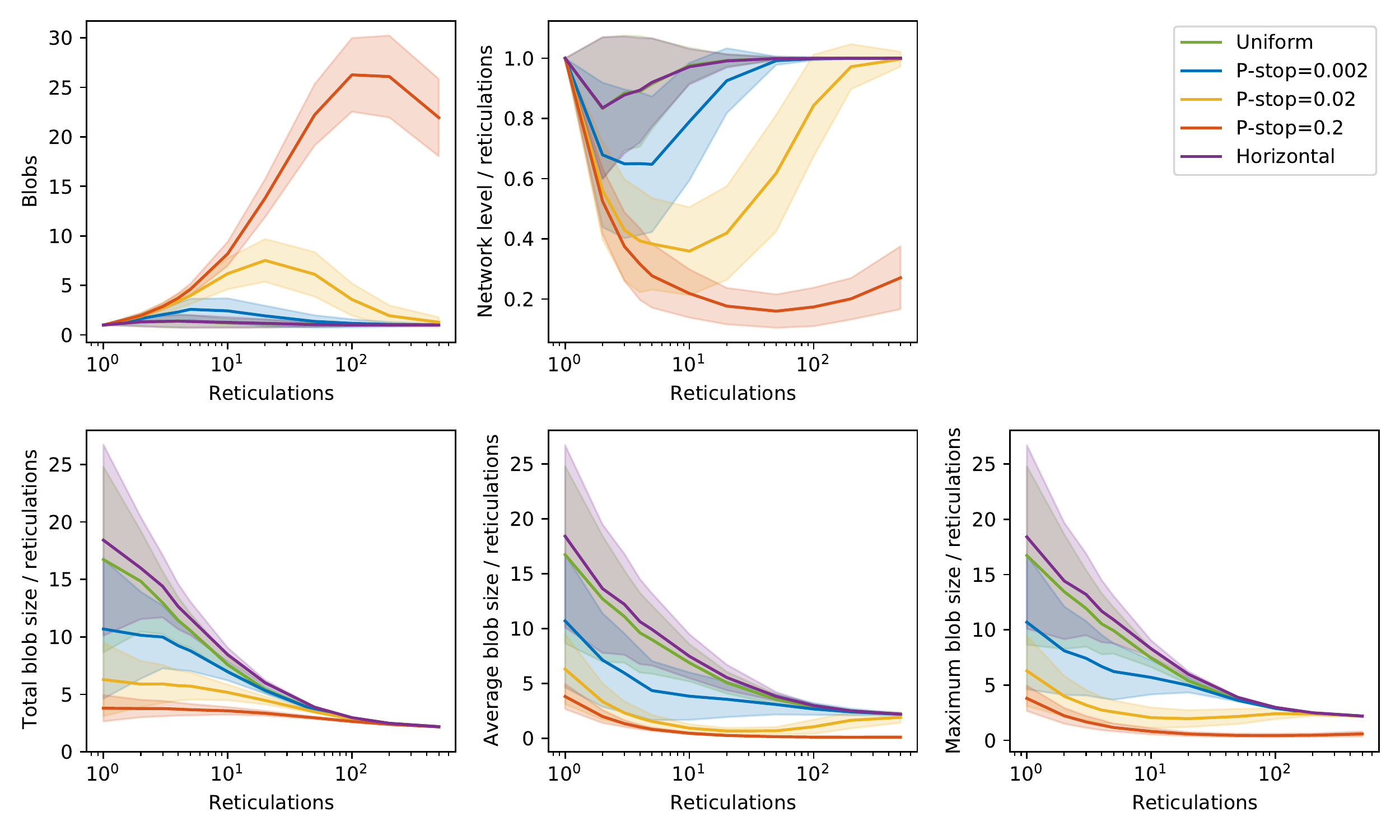}
    \caption{Dependence of the number and size and level of blobs dependent on the reticulation number for networks with 100 leaf nodes. Top: Beta-splitting beta=0.0 networks, Bottom: Beta-splitting beta=-1.5 networks.}
    \label{fig:BlobsBetaSplitByBeta}
\end{figure}

\newpage
\subsection{Heath Parameters and Extinction}\label{sec:HeathExt}
In the main text of this paper, we have not considered extinction. Even though our implementation of the Heath network generator has an option to add extinction. In Figures~\ref{fig:Classes_Heath},~\ref{fig:Balance_Heath},~\ref{fig:Cherries_Heath}, and~\ref{fig:Blobs_Heath}, we show the effect of extinction on the profiles, when using the parameters for the Heath network generator as detailed in Table~\ref{tab:HeathParameters}. We have two methods of implementing extinction: simply generating a network where the desired number of taxa consists of both extinct and extant taxa; or generating a network with extinct taxa, and removing them to obtain a network with the desired number of extant taxa. When extinct networks are removed, Table~\ref{tab:Classes} does not show reticulation numbers at which half of the generated networks are orchard or tree-based. This is because, for the reticulation numbers we could test for this model, none gave a fraction of less than half for these classes (Figure~\ref{fig:Classes_Heath}). Perhaps, this is because our networks have a small number of extinct taxa, and removing extinct taxa is necessary to produce networks that are not orchard or tree-based when adding reticulations as hybrid speciation.

\begin{figure}[h!]
    \centering
    \includegraphics[width=.7\textwidth]{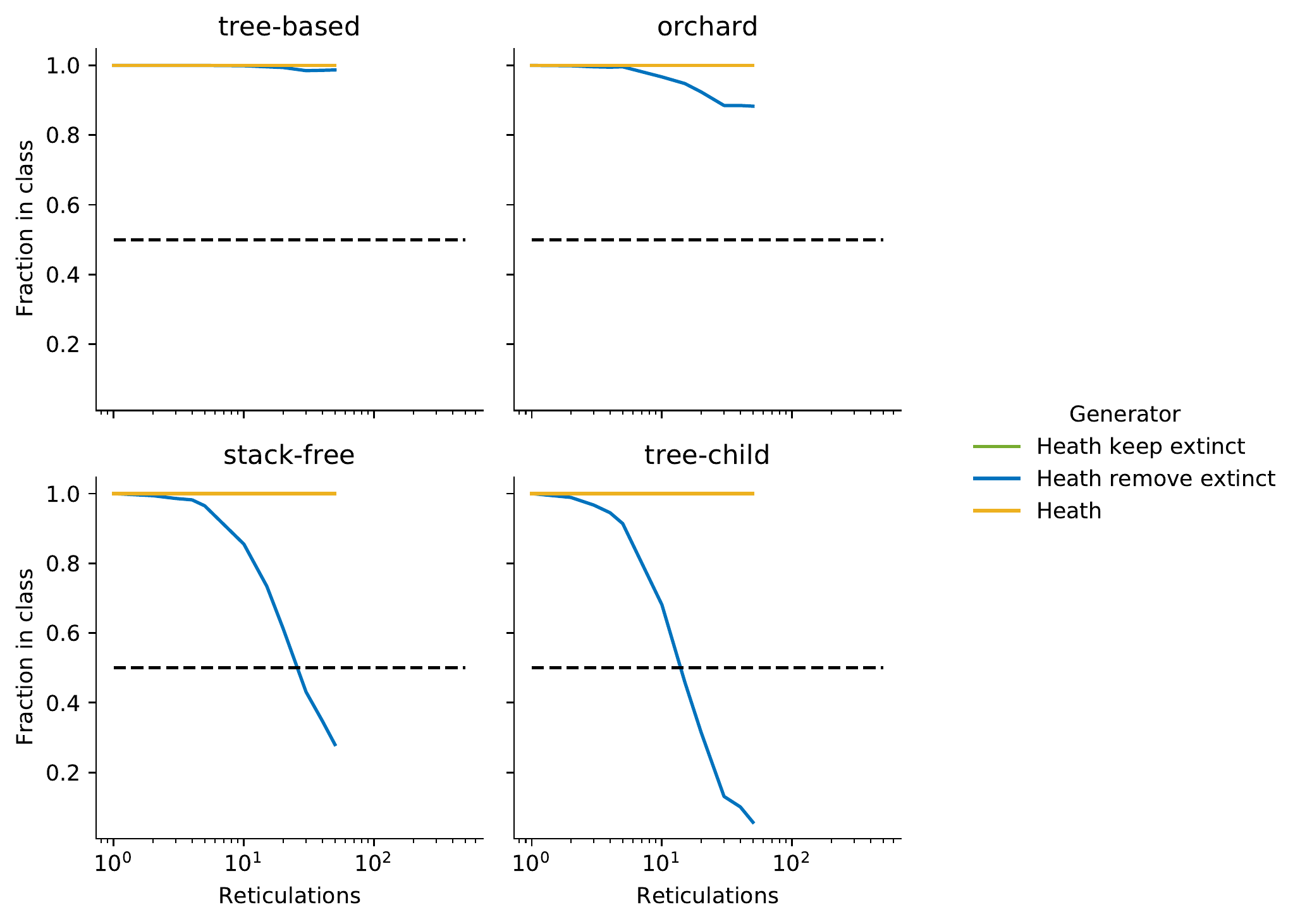}
    \caption{Proportions of generated networks with 100 leaves in different classes for the Heath network generator with different methods of handling extinction. The lines for ``Heath'' and ``Heath keep extinct'' are exactly the same, as these generators only produce tree-child networks; only when extinct lineages are removed, the Heath network generator can produce networks that are not tree-child, and even not tree-based. }
    \label{fig:Classes_Heath}
\end{figure}

\begin{figure}[h!]
    \centering
    \includegraphics[width=.8\textwidth]{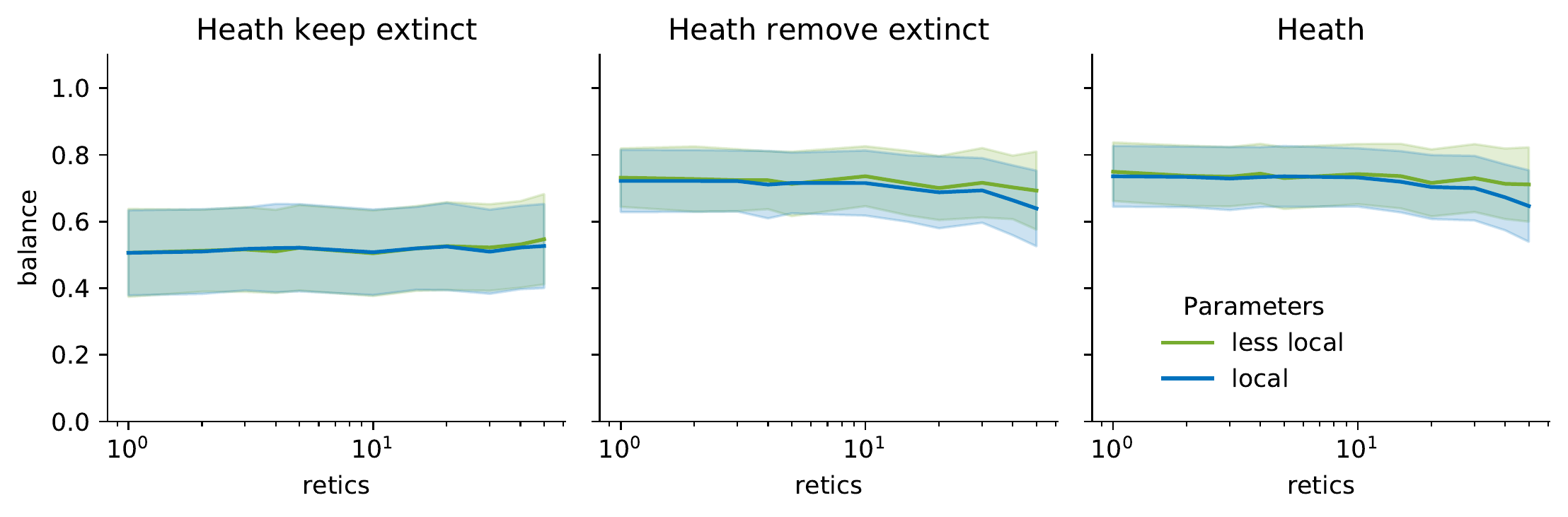}
    \caption{The balance for networks with 100 leaves generated by the Heath network generator for different methods of handling extinction.}
    \label{fig:Balance_Heath}
\end{figure}

\begin{figure}[h!]
    \centering
    \includegraphics[scale = 0.6]{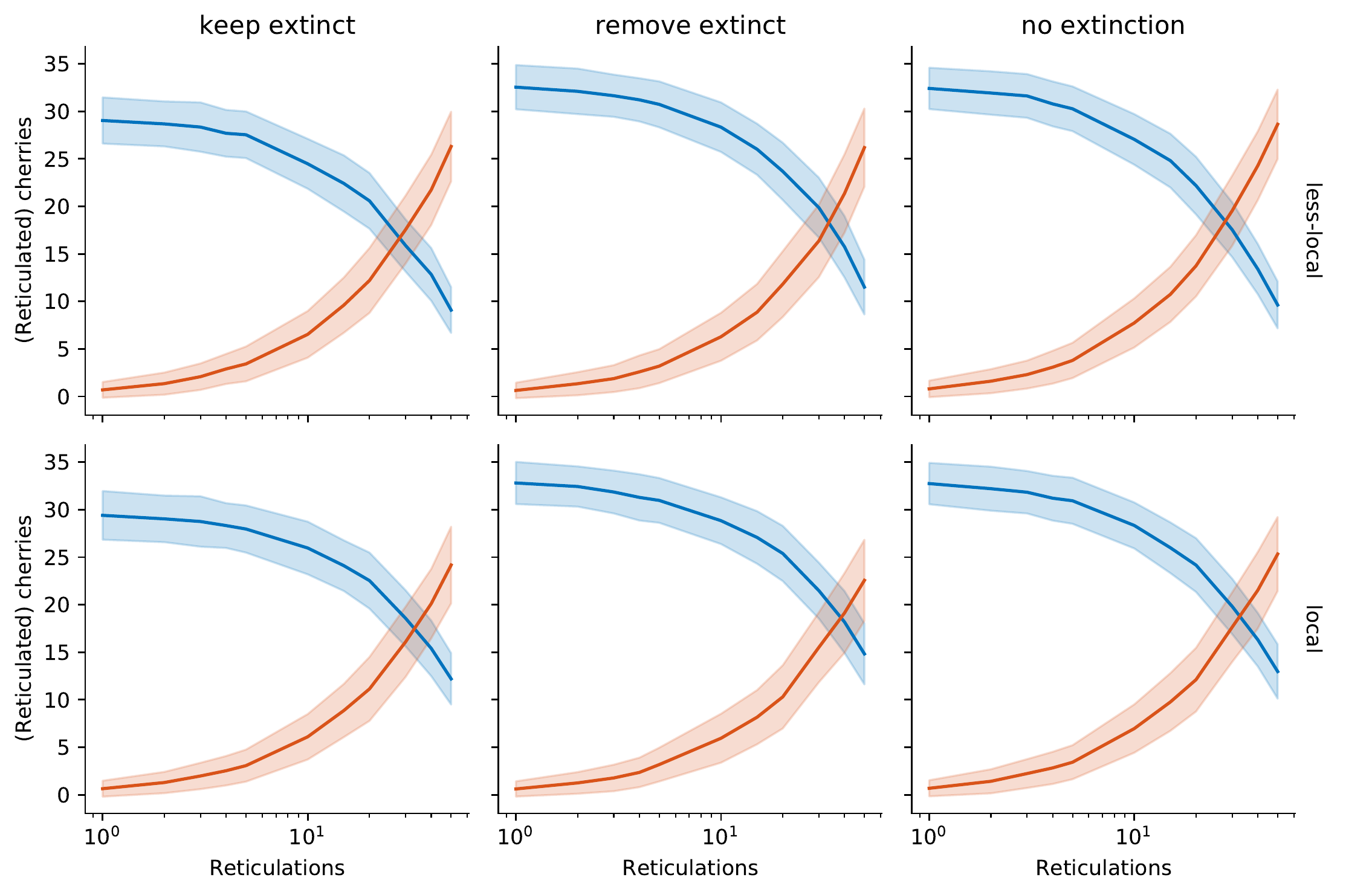}
    \caption{Dependence of the number of cherries (blue) and reticulated cherries (orange) for networks with 100 leaf nodes in the Heath network generator.}
    \label{fig:Cherries_Heath}
\end{figure}

\begin{figure}[h!]
    \centering
    \includegraphics[scale = 0.52]{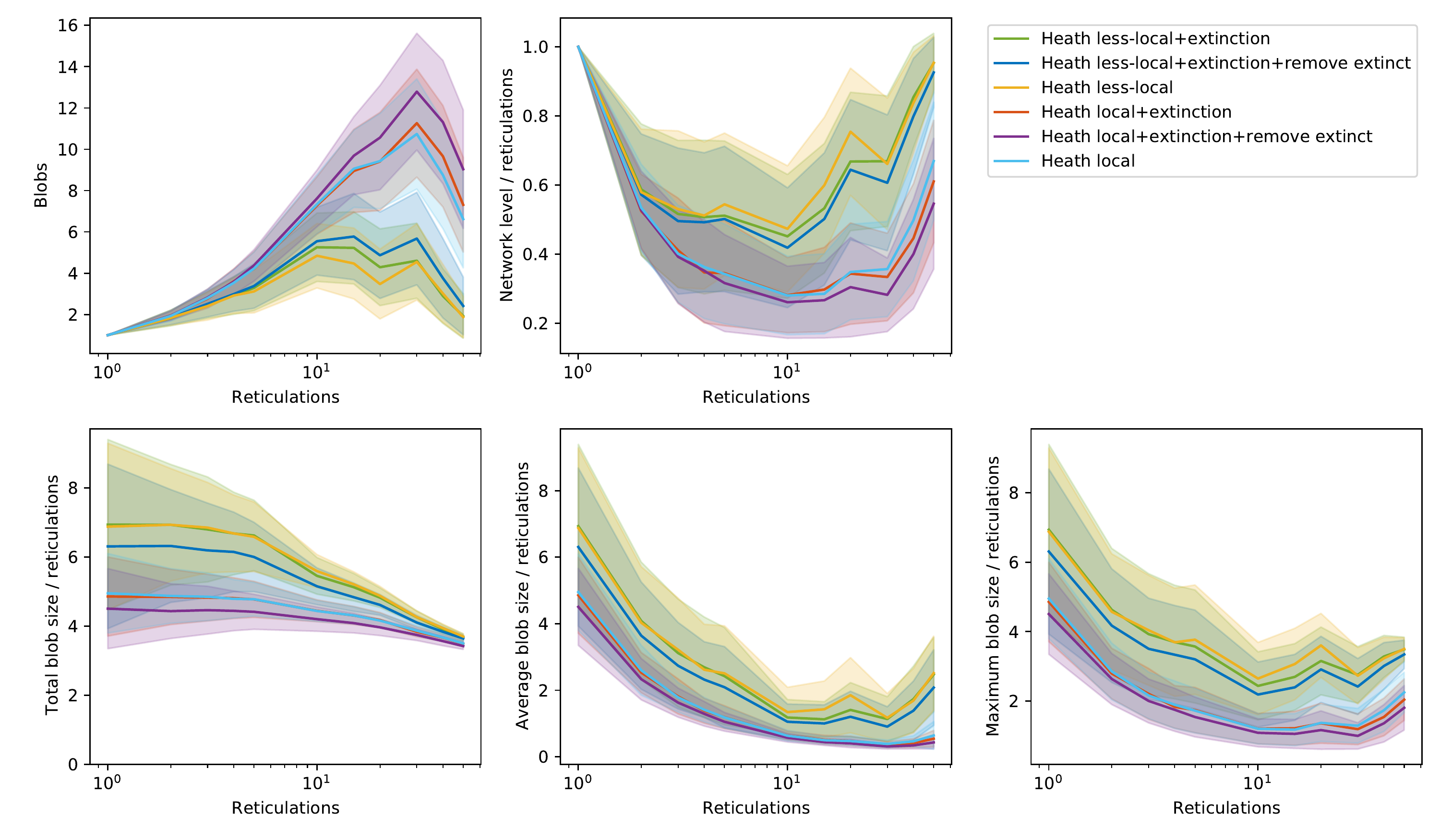}
    \caption{Dependence of the number and size and level of blobs dependent on the reticulation number for the Heath network generator.}
    \label{fig:Blobs_Heath}
\end{figure}

\begin{sidewaystable}[h!]
    \caption{Parameters used in the Heath network generator}
    \label{tab:HeathParameters}
    \centering
    \tiny
    \begin{tabular}{|l|c|c|c|c|c|c|c|c|c|c|c|c|c|c|c|c|c|c|c|}
    \hline
    Extinction                              & \multicolumn{6}{|c|}{No}                                            & \multicolumn{6}{|c|}{Yes - remove}                                  & \multicolumn{6}{|c|}{Yes - keep}   \\ \hline
    Reticulations                           & \multicolumn{2}{|c|}{$[1,2,3,4,5]$} & \multicolumn{2}{|c|}{$[10,15,20]$}  & \multicolumn{2}{|c|}{$[30,40,50]$}  & \multicolumn{2}{|c|}{$[1,2,3,4,5]$} & \multicolumn{2}{|c|}{$[10,15,20]$}  & \multicolumn{2}{|c|}{$[30,40,50]$} & \multicolumn{2}{|c|}{$[1,2,3,4,5]$} & \multicolumn{2}{|c|}{$[10,15,20]$}  & \multicolumn{2}{|c|}{$[30,40,50]$} \\ \hline
    Locality                                & local & less local & local & less local & local & less local & local & less local & local & less local & local & less local & local & less local & local & less local & local & less local \\ \hline
    update shape                            & 3.0 & 3.0 & 3.0 & 3.0 & 3.0 & 3.0 & 3.0 & 3.0 & 3.0 & 3.0 & 3.0 & 3.0 & 3.0 & 3.0 & 3.0 & 3.0 & 3.0 & 3.0  \\
    speciation rate mean                    & 1.0 & 1.0 & 1.0 & 1.0 & 1.0 & 1.0 & 1.0 & 1.0 & 1.0 & 1.0 & 1.0 & 1.0 & 1.0 & 1.0 & 1.0 & 1.0 & 1.0 & 1.0  \\
    speciation rate shape                   & 5.0 & 5.0 & 5.0 & 5.0 & 5.0 & 5.0 & 5.0 & 5.0 & 5.0 & 5.0 & 5.0 & 5.0 & 5.0 & 5.0 & 5.0 & 5.0 & 5.0 & 5.0  \\
    extinction rate mean                    & 0.0 & 0.0 & 0.0 & 0.0 & 0.0 & 0.0 & 0.5 & 0.5 & 0.5 & 0.5 & 0.5 & 0.5 & 0.5 & 0.5 & 0.5 & 0.5 & 0.5 & 0.5  \\
    extinction rate shape                   & 0.0 & 0.0 & 0.0 & 0.0 & 0.0 & 0.0 & 5.0 & 5.0 & 5.0 & 5.0 & 5.0 & 5.0 & 5.0 & 5.0 & 5.0 & 5.0 & 5.0 & 5.0  \\
    $h_l$                                   & 0.3 & 1.0 & 0.3 & 1.0 & 0.3 & 1.0 & 0.3 & 1.0 & 0.3 & 1.0 & 0.3 & 1.0 & 0.3 & 1.0 & 0.3 & 1.0 & 0.3 & 1.0  \\
    $h_r$                                   & 1.5 & 5.0 & 1.5 & 5.0 & 1.5 & 5.0 & 1.5 & 5.0 & 1.5 & 5.0 & 1.5 & 5.0 & 1.5 & 5.0 & 1.5 & 5.0 & 1.5 & 5.0  \\
    $h_{lr}$                                & 0.08 & 0.015 & 0.3 & 0.05 & 0.6 & 0.125 & 0.06 & 0.01 & 0.25 & 0.05 & 0.6 & 0.11 & 0.08 & 0.015 & 0.3 & 0.075 & 0.75 & 0.15 \\
    $h_{rr}$                                & 0.0 & 0.0 & 0.0 & 0.0 & 0.0 & 0.0 & 0.0 & 0.0 & 0.0 & 0.0 & 0.0 & 0.0 & 0.0 & 0.0 & 0.0 & 0.0 & 0.0 & 0.0  \\ \hline 
    \end{tabular}
\end{sidewaystable}

\end{document}